\documentclass[a4paper,fleqn]{cas-sc}

\usepackage[numbers]{natbib}
\usepackage{amsthm}
\usepackage{bm}
\usepackage{listings}
\usepackage{comment}
\usepackage[section]{placeins}
\newcommand{\pkgname}{PyBlockEncode}
\newcommand{\libname}{pyblockencode}
\newcommand{\repourl}{https://github.com/UW-ERSL/PyBlockEncode}

\newcommand{\Ublk}{\mathcal{U}}
\newcommand{\Sc}{S_c}
\newcommand{\Scd}{S_c^\dagger}
\newcommand{\kron}{\otimes}
\newcommand{\norm}[1]{\ensuremath{\|#1\|}}

\newcommand{\Rchi}{R_{\chi}}
\newcommand{\Ksc}{\mathbf{K}_{2D,\mathrm{sc}}}
\newcommand{\Kel}{\mathbf{K}_{2D,\mathrm{el}}}
\newcommand{\asc}{\alpha_{2D,\mathrm{sc}}}
\newcommand{\ael}{\alpha_{2D,\mathrm{el}}}
\newtheorem{proposition}{Proposition}

\definecolor{codegray}{gray}{0.95}
\definecolor{codecomment}{rgb}{0.0,0.45,0.0}
\definecolor{codekw}{rgb}{0.0,0.0,0.6}
\begin{document}

\let\WriteBookmarks\relax

\shorttitle{Block encodings for periodic two-phase finite element operators}
\shortauthors{K. Suresh}

\title[mode = title]{Quantum Block Encodings for Periodic Two-Phase Finite
  Element Operators: 2D Poisson and 2D Elasticity}

\author[1]{Krishnan Suresh}[
   orcid=0000-0002-9688-9697
  ]
\cormark[1]
\ead{ksuresh@wisc.edu}
\credit{Conceptualization, Methodology, Software, Validation,
  Writing -- original draft}

\affiliation[1]{organization={University of Wisconsin--Madison,
                              Department of Mechanical Engineering},
            city={Madison},
            state={WI},
            postcode={53706},
            country={USA}}

\cortext[1]{Corresponding author}

\begin{abstract}
Quantum algorithms require embedding linear operators into unitaries using block encodings. The costs associated with these block encodings determine whether a quantum speedup is achieved. Efficient \emph{shift decomposition} encodings have been proposed for 2D homogeneous scalar operators. Here, we present exact block encodings for 2D homogeneous elasticity, 2D two-phase bilinear Poisson, and 2D two-phase elasticity periodic finite-element operators. 

For 2D homogeneous elasticity, the resulting linear combination of unitaries (LCU)
has $L = 17$ terms for every Poisson ratio $\nu$ and every mesh resolution, and
the subnormalization is the closed form
$\alpha = E(33+\nu)/\bigl[6(1-\nu^{2})\bigr]$, exceeding $\|\mathbf{K}\|_{2}$ by the
resolution-independent factor $(33+\nu)/24$.

To address two-phase periodic microstructures, we introduce two oracles with
distinct roles: a node-to-element oracle, built from cyclic shifts, that carries
a nodal index to each of the four incident element indices, and a microstructure-specific material
oracle that marks phase membership with a sign. This yields 25 and 57 LCU terms for two-phase Poisson and elasticity, respectively; the latter reduces to 49 at $\nu = 1/3$. The term counts are independent of mesh resolution, volume fraction, and phase contrast. The subnormalizations, provided in closed form, are independent of mesh resolution and volume fraction. In the scalar case, the tight bound $\alpha = \|\mathbf{K}\|_\infty$ is achieved whenever a node lies entirely in the stiffer phase.  In the elasticity case that bound is not attained, since the shear coupling
contributes entries of both signs to a row. An open-source implementation is available at \url{\repourl}.
\end{abstract}



\maketitle

\section{Introduction}
\label{sec:intro}
 
Quantum algorithms for linear systems, differential equations, and matrix
functions promise substantial speedups for large-scale computational
problems~\cite{harrow2009quantum, childs2017quantum,costa2022optimal, gilyen2019quantum,berry2015simulating}.
 Modern algorithmic formulations require the underlying operator
$\mathbf{A}$ to be embedded in a larger unitary $\Ublk$ as a \emph{block
encoding}: an $(a+n)$-qubit unitary whose top-left $2^n \times 2^n$ block is
$\mathbf{A}/\alpha$ for some subnormalization
$\alpha \geq \norm{\mathbf{A}}_2$.
Constructing that encoding is the first step of every modern pipeline. The block encoding cost determines whether a quantum speedup is achieved. A general-purpose encoder working from the entries of $\mathbf{A}$ alone, such as FABLE~\cite{camps2022fable}, costs $\mathcal{O}(4^{n})$ gates at
$\alpha = 2^{n}\max_{ij}|A_{ij}|$ for a $2^{n} \times 2^{n}$ matrix, and the
exponential overhead is unavoidable without structure~\cite{kuklinski2025structure}.

Operators in computational mechanics over regular meshes typically have structure. A finite element
discretization on a regular grid is sparse, assembled from one small element
matrix, and expressible in terms of cyclic shifts. A substantial literature, reviewed in Section~\ref{sec:litreview}, uses this structure for scalar 1D and 2D homogeneous operators.

In this work, we construct exact shift-decomposition-based block encodings of three operators: 2D homogeneous elasticity, 2D two-phase bilinear Poisson, and 2D two-phase elasticity. 

Section~\ref{sec:litreview} provides a brief technical background and reviews
the literature. Section~\ref{sec:StructuredEncoding} recalls the
shift-decomposition encoding of the periodic homogeneous Poisson cell, and
Section~\ref{sec:elasticity} derives the same for elasticity.
Section~\ref{sec:microstructure} introduces the two-phase microstructure, the
node-to-element mapping and the material oracle.
Section~\ref{sec:twophase_poisson} and Section~\ref{sec:twophase_elasticity}
generalize the two encodings to two phases. Section~\ref{sec:encoding}
describes the circuit and summarizes the resources,
and Section~\ref{sec:implementation} documents the open-source implementation
and the reproduction of every table reported here.
Appendix~\ref{sec:circuits} draws the four circuits at $m = 2$.

\section{Background}
\label{sec:litreview}
 
\subsection{Block Encoding and LCU}
\label{sec:be}
 
A unitary $\Ublk$ acting on $a + n$ qubits is an $(\alpha, a,
\varepsilon)$-block encoding \cite{gilyen2019quantum} of an $n$-qubit operator $\mathbf{A}$ if
\begin{equation}
\label{eq:be}
\Bigl\|\mathbf{A} - \alpha\,
  \bigl(\langle 0|^{\otimes a} \otimes I\bigr)\,\Ublk\,
  \bigl(|0\rangle^{\otimes a} \otimes I\bigr)
\Bigr\|_2 \le \varepsilon.
\end{equation}
We place the system register on the least-significant qubits, so the encoded
block is the contiguous top-left corner $\Ublk[{:}2^n, {:}2^n]$. Setting
$\varepsilon = 0$ gives an exact encoding. 
 
The canonical construction assembles $\Ublk$ as a \emph{linear combination of
unitaries} (LCU)~\cite{childs2012hamiltonian,berry2015simulating}. Given
\begin{equation}
\label{eq:Pauli}
\mathbf{A} = \sum_{k=0}^{L-1} c_k U_k
\end{equation}
 with unitaries $U_k$ and real coefficients $c_k$, the subnormalization is $\alpha = \sum_k |c_k|$. For
symmetric $\mathbf{A}$, and whenever every $U_k$ is a \emph{generalized
permutation}, carrying one nonzero entry per row,
\begin{equation}
\label{eq:norm-chain}
\norm{\mathbf{A}}_2
  \;\le\; \norm{\mathbf{A}}_\infty = \max_i \sum_j |A_{ij}|
  \;\le\; \alpha = \sum_k |c_k| .
\end{equation}
Every unitary used in this paper is of that form: each $U_k$ has at most
one entry in a given row, so no cancellation between terms can push $\alpha$
below the largest absolute row sum, and $\norm{\mathbf{A}}_\infty$ is a floor
on what the decomposition can attain. 

The subnormalization $\alpha$ controls both the post-selection success
probability $p = \norm{\mathbf{A}x}^2/\alpha^2$ at $\norm{x}_2 = 1$, and the
query complexity of downstream routines. The circuit then has three stages: \textsc{prep} loads
$\sqrt{|c_k|/\alpha}$ into the ancilla register, \textsc{select} applies $U_k$
controlled on the ancilla state $|k\rangle$, and \textsc{unprep} uncomputes
the ancilla, with the sign of $c_k$ absorbed into $U_k$ or carried by a second
preparation. On the $|0\rangle$-ancilla subspace,
\begin{equation}
\label{eq:lcu-block}
\bigl(\langle 0|_{\mathrm{anc}} \otimes I\bigr)\,\Ublk\,
  \bigl(|0\rangle_{\mathrm{anc}} \otimes I\bigr) = \mathbf{A}/\alpha,
\quad
\alpha = \sum_k |c_k| .
\end{equation}
 
Consequently, three quantities determine the cost of block-encoding:
\begin{enumerate}\itemsep1pt
\item The subnormalization $\alpha = \sum_k |c_k|$, which affects the query
  complexity linearly and the post-selection
  probability quadratically.
\item The Toffoli count of \textsc{select}, which, for a generic implementation,
  is $\mathcal{O}(L)$ \cite{babbush2018encoding}.
\item The arbitrary-angle rotation count of \textsc{prep}, which dominates the
  fault-tolerant cost because each rotation costs
  $\mathcal{O}(\log 1/\varepsilon_{\mathrm{syn}})$ $T$ gates at
  precision $\varepsilon_{\mathrm{syn}}$, while everything else is Clifford
  and Toffoli.
\end{enumerate}
The \textsc{prep} register must address all $L$ terms, so its width satisfies
$n_{\mathrm{prep}} \ge \lceil \log_2 L\rceil$. A dense \textsc{prep} on
$n_{\mathrm{prep}}$ qubits uses $2^{n_{\mathrm{prep}}} - 1$ rotations, whether
or not every index carries amplitude.

\subsection{Pauli Encodings}
\label{sec:generic}
 
The common generic approach is LCU expansion in the Pauli basis \cite{suresh2026appliedqc}. The Pauli term count, verified numerically for $m = 2$ to $6$, is
\begin{equation}
\label{eq:pauli-counts}
L_{\mathrm{Pauli}} =
\begin{cases}
 \tfrac34 2^m, & \text{1D Poisson},\\
 \tfrac{9}{16} 4^m, & \text{2D Poisson (FE)},\\
 \tfrac98 4^m - 1, & \text{2D elasticity},\\
 (2^{m+1} - 1)^2, & \text{2D Poisson, 2-phase},\\
  \tfrac{19}{2} 4^m - 11 \cdot 2^m + 1, & \text{2D elasticity, 2-phase},
\end{cases}
\end{equation}
for $m$ qubits per spatial direction. The two 2-phase counts hold at every
contrast, but they are specific to the microstructure of
Figure~\ref{fig:inclusion}; other microstructures may give different counts. The
subnormalization grows quadratically in $m$ for both homogeneous 2D cases. Considerable effort has gone into computing Pauli coefficients faster, through
tensorized recursion~\cite{hantzko2024tensorized}, tree
traversal~\cite{koska2024tree}, symmetry-restricted
enumeration~\cite{pesce2021h2zixy} and the fast Walsh--Hadamard
transform~\cite{georges2025pauli}.

Beyond standard Pauli decompositions, several alternative matrix encoding strategies have been proposed. Lu et al.~\cite{lu2026memory} benchmarked multiple encoding approaches: parallel Walsh--Hadamard expansions, standard LCU methods, and an SVD-based singular value dilation that reduces any real square matrix to an exact two-term LCU. The SVD dilation attains $\alpha = \norm{\mathbf{A}}_2$, but its two factors are dense rather than generalized permutations.  Other general block-encoding frameworks, such as state preparation pairs or oracle-based linear combinations of unitaries~\cite{gilyen2019quantum}, similarly shift the computational complexity to deep state-preparation circuits or unmapped oracle operations unless specific structural assumptions are made. Xu and Hu~\cite{xu2025decomposition} remove the operator encoding entirely; instead they vectorize $\mathbf{K}$ and evaluate the variational cost functions as inner
products through a swap test; the expense moves to state preparation, which they state as an open problem.
 
Kuklinski et al.~\cite{kuklinski2025structure} make the premise of this paper
precise: \emph{efficient block encodings require structure}.

\section{Structured Block Encoding}
\label{sec:StructuredEncoding}
 
On a regular grid, a finite difference (FD) or finite element (FE)
discretization of the Poisson or elasticity operator produces a matrix that is
\emph{sparse} and \emph{translation invariant}: every node carries the same
stencil, built from a \emph{small set of repeated entries}.
\emph{Structured block encoding} is the collection of techniques that use
this structure so that the gate and ancilla counts scale as
$\mathcal{O}(\mathrm{poly}(m))$ rather than $\mathcal{O}(2^{m})$, and the
subnormalization stays constant, for an $m$-qubit operator.
 
Several families of structured block encoding have been proposed. These
include \emph{shift decomposition}, where we write the operator in closed form
as a short sum of \emph{shift} operators and use that sum in an
LCU~\cite{childs2012hamiltonian, kharazi2025explicit, sturm2025efficient,
alkadri2025quantum, hogancamp2026linear}. A closely related recursion expands the Dirichlet stencil in $\mathcal{O}(m)$ terms built from the ladder operators
$\bm\sigma_\pm = (\bm X \pm i\bm Y)/2$~\cite{liu2021variational};
Gnanasekaran and Surana~\cite{gnanasekaran2026efficient} generalize this to a
basis of such operators and recover the decomposition by unitary completion,
at $\mathcal{O}(m)$ terms for the one-dimensional Poisson operator and a
\textsc{select} cost that grows with the term count.  The
\emph{sparse-access} approach places the nonzero values of a structured operator
by \emph{addressing} them directly, through a pair of oracles~\cite{camps2024explicit, sunderhauf2024block, Danz_2026}. Finally, the \emph{element-decomposition} method, which is specific to FE, assembles the
global stiffness matrix from small element matrices~\cite{arora2025implementation, holscher2026end}.

This paper focuses on the \emph{shift decomposition} approach and addresses two limitations in the literature. The first is the scalar restriction. Shift decomposition has been developed for scalar
operators, and the two constructions that address linear elasticity take a different approach: H\"olscher et al.~\cite{holscher2026end} take an LCU over elements, so the term count equals the element count and the
subnormalization grows with the mesh, and Danz et al.~\cite{Danz_2026} build
the stiffness and mass matrices of an elastic structure from general
fixed-point arithmetic oracles. Neither yields a decomposition
whose length is independent of resolution.

The second limitation is the homogeneity restriction. Block encodings of
variable-coefficient operators do exist: Deiml and
Peterseim~\cite{deiml2024quantum} admit an arbitrary spatially varying
coefficient but leave its oracle abstract and the subnormalization 
grows with the ellipticity contrast; Pechan et al.~\cite{pechan2026block}
construct a heterogeneous three-dimensional Poisson encoding from the sparse
local structure of the discretized operator, at a subnormalization bounded
by the largest conductivity in the cell; Yano and
Sato~\cite{yano2026quantum} represent a piecewise coefficient profile by a
controlled Fourier series, which is approximate at a material interface. What
none of them provides is a linear combination of fixed length in closed form,
over a material oracle realized exactly.

Table~\ref{tab:prior} compares the scope of the proposed work against the existing literature.
\begin{table}[H]
\centering
\footnotesize
\setlength{\tabcolsep}{6pt}
\renewcommand{\arraystretch}{1.3}
\caption{Prior block encodings and the present construction.
S scalar, V vector; FD finite difference, FE finite element; Dim spatial
dimension, $d$ any dimension; U uniform, H heterogeneous, 2 two-phase;
Y/N whether the work names a public repository.
$\mathcal{O}(1)$ means constant in $N$; the term count of a shift construction
grows linearly with $d$ for finite differences and as $3^d$ for tensor-product
finite elements.
$\dagger$~also treats non-periodic boundary conditions;
$\ddagger$~holds for constant coefficients;
$\S$~the framework is stated for any $d$, but the stiffness analysis is carried
out in one dimension;
$\star$~the two SVD factors are dense and unstructured, and no efficient circuit
for them is given.
\emph{n/a} marks a quantity that does not apply, a dash one the cited work does
not report. Every $\alpha$ is quoted at unit element length.}
\label{tab:prior}
\begin{tabular}{@{}l c c c c l c c c@{}}
\toprule
work & Op & FE/FD & Dim & Mat & Approach & $L$ & $\alpha$ & code \\
\midrule
Kharazi (2025)~\cite{kharazi2025explicit}
  & S & FD & $d$ & U & shift$^\dagger$
  & $\mathcal{O}(1)$ & $\mathcal{O}(1)$ & Y \\
Sturm (2025)~\cite{sturm2025efficient}
  & S & FD & $d$ & U & shift
  & $\mathcal{O}(1)$ & $\mathcal{O}(1)$ & N \\
Hogancamp (2026)~\cite{hogancamp2026linear}
  & S & FD & $d$ & U & shift$^\dagger$
  & $\mathcal{O}(1)$ & $\mathcal{O}(1)$ & N \\
Mahmud (2026)~\cite{mahmud2026moment}
  & S & FD & $d$ & U & shift $+$ optimality
  & $\mathcal{O}(1)$ & $\mathcal{O}(1)$ & N \\
Gnanasekaran (2026)~\cite{gnanasekaran2026efficient}
  & S & FD & $1^\S$ & U & ladder $+$ completion
  & $\mathcal{O}(m)$ & --- & N \\
Alkadri (2025)~\cite{alkadri2025quantum}
  & S & FE & $d$ & U, H & element LCU
  & $\mathcal{O}(1)^\ddagger$ & $\mathcal{O}(1)^\ddagger$ & N \\
\midrule
Deiml (2025)~\cite{deiml2024quantum}
  & S & FE & $d$ & H & FE $+$ BPX, oracle
  & n/a & $\propto$ contrast & Y \\
Pechan (2026)~\cite{pechan2026block}
  & S & FD & 3 & H & sparse access
  & n/a & $12k_{\max}$ & Y \\
Yano (2026)~\cite{yano2026quantum}
  & S & FD & $d$ & H & diagonal, Fourier
  & n/a & assumed & N \\
\midrule
Danz (2026)~\cite{Danz_2026}
  & V & FE & $1^\S$ & U, H & arithmetic oracle
  & n/a & --- & N \\
Arora (2025)~\cite{arora2025implementation}
  & S & FE & 1 & H & VQLS; element 
  & $N+2$ & n/a & N \\
Lu (2026)~\cite{lu2026memory}
  & S & any & 2 & n/a & coh. VQLS; SVD
  & 2$^\star$ & 1$^\star$ & N \\
H\"olscher (2026)~\cite{holscher2026end}
  & V & FE & 2 & 2 & element $+$ Grover
  & $\mathcal{O}(n_{\mathrm{el}})$ & $n_{\mathrm{el}}\delta$ & Y \\
\midrule
This work
  & S, V & FE & 2 & 2 & shift $+$ reflections
  & $\mathcal{O}(1)$ & $\mathcal{O}(1)$ & Y \\
\bottomrule
\end{tabular}
\end{table}

\subsection{1D Homogeneous Periodic Poisson}
\label{sec:1DPoisson}

Shift decomposition is easiest to see on a 1D operator.
For the 1D Poisson problem on $N = 2^m$ \emph{periodic}
nodes, discretized using linear finite elements of unit element length, leads to the stiffness matrix
\begin{equation}
\label{eq:K1D}
\mathbf{K}_{1D} = \begin{bmatrix}
 2 & -1 &  0 & \cdots & -1 \\
-1 &  2 & -1 & \cdots &  0 \\
\vdots &  & \ddots &  & \vdots \\
-1 &  0 & \cdots & -1 & 2
\end{bmatrix},
\end{equation}
where the corner entries carry the periodic wrap. This can be expressed compactly as
\begin{equation}
\label{eq:K1}
\mathbf{K}_{1D} \;=\; \mathrm{circ}(-1,\,2,\,-1),
\end{equation}
where $\mathrm{circ}(a,b,c)$ carries $a$ on the sub-diagonal, $b$ on the
diagonal and $c$ on the super-diagonal, each off-diagonal band wrapping into
the opposite corner.

Let $\Sc$ be the cyclic shift
\begin{equation}
\label{eq:Sc-1D}
\Sc = \begin{bmatrix}
0 & 0 & 0 & \cdots & 1 \\
1 & 0 & 0 & \cdots & 0 \\
0 & 1 & 0 & \cdots & 0 \\
\vdots &  & \ddots &  & \vdots \\
0 & 0 & \cdots & 1 & 0
\end{bmatrix},
\end{equation}
so that $\Sc: |j\rangle \mapsto |(j{+}1) \bmod N\rangle$. It is a permutation
matrix, hence unitary, and its adjoint $\Scd$ is the backward shift. Thus
Equation~\eqref{eq:K1D} can be expressed as
\begin{equation}
\label{eq:K1D-shift}
\mathbf{K}_{1D} = 2I - \Sc - \Scd .
\end{equation}
This is an exact LCU, and hence a block encoding: three unitaries,
coefficients known in closed form, and a subnormalization $\alpha = 4$ that
does not grow with $N$. The cost is one increment modulo $N$ of depth
$\mathcal{O}(m)$~\cite{draper2000addition,camps2024explicit}; the controlled
ripple-carry increment used here costs $2m-2$ Toffoli, $m$ CNOT and $m-1$ clean ancillas
that are returned to $|0\rangle$ and reused by every shift in the
circuit~\cite{gidney2018halving, cuccaro2004ripple}. 

Table~\ref{tab:poisson1d-compare} compares the cost of encoding $\mathbf{K}_{1D}$ using shift decomposition against an exact Pauli expansion of the same operator at $m = 12$. The gap widens in 2D and for heterogeneous cells. For the 1D cell, the chain $\norm{\mathbf{K}}_2 \le \norm{\mathbf{K}}_\infty \le \alpha$ collapses to
equality, so the shift decomposition is optimal.

\begin{table}[H]
\centering
\footnotesize
\setlength{\tabcolsep}{6pt}
\renewcommand{\arraystretch}{1.2}
\caption{Periodic 1D finite element Poisson cell at $m = 12$ qubits,
$4{,}096$ degrees of freedom. The shift column is measured; the Pauli column
is exact in closed form, $L = \tfrac34 2^m$, $\alpha = m+2$ and
$p_{\mathrm{succ}} = 16/(m+2)^2$, since $\norm{\mathbf{K}_{1D}}_2 = 4$ at
every even $N$. This is the one cell in which the comparison is tabulated;
the four 2D cells follow the same pattern more sharply, at the term counts of
Equation~\eqref{eq:pauli-counts}.}
\label{tab:poisson1d-compare}
\begin{tabular}{@{}lrr@{}}
\toprule
quantity & Pauli & shift \\
\midrule
terms $L$                          & $3{,}072$ & $3$ \\
\textsc{prep} qubits               & $12$      & $2$ \\
\textsc{prep} slots               & $4{,}096$ & $4$ \\
\textsc{select} Toffoli            & $3{,}072$ & $48$ \\
$\alpha$                           & $14.0$    & $4.000$ \\
$p_{\mathrm{succ}}$                & $0.082$   & $1.00$ \\
\bottomrule
\end{tabular}
\end{table}

\subsection{2D Homogeneous Periodic Poisson}
\label{sec:2DPoisson}

For a 2D periodic Poisson problem consisting of a $N \times N$ grid of bilinear $Q_4$ elements, the finite element operator separates into 1D factors~\cite{deville2002high, van2000ubiquitous, alkadri2025quantum},
\begin{equation}
\label{eq:K2D-sep}
\Ksc = \mathbf{K}_{1D} \kron \mathbf{M}_{1D}
                + \mathbf{M}_{1D} \kron \mathbf{K}_{1D},
\end{equation}
where $\mathbf{M}_{1D} = \tfrac{1}{6}\,\mathrm{circ}(1,\,4,\,1)$ is the consistent mass matrix. Throughout this paper, the left factor of a Kronecker
product acts on the $x$ register, the next on the $y$ register; $\kron$ thus
lists registers from least to most significant, the reverse of the matrix
Kronecker product. Both factors are three-band
circulants, so both admit the shift decomposition of
Section~\ref{sec:1DPoisson},
\begin{equation}
\label{eq:M1-shift}
\mathbf{M}_{1D} = \tfrac{4}{6}I + \tfrac{1}{6}\Sc + \tfrac{1}{6}\Scd ,
\end{equation}
and Equation~\eqref{eq:K2D-sep} expands over the nine unitaries
$\{I, \Sc, \Scd\}^{\kron 2}$. This approach to the periodic scalar cell is due to
Sturm and Schillo~\cite{sturm2025efficient} and, for tensor-product finite
elements, to Qu-FEM~\cite{alkadri2025quantum}.

Merging the coefficients that land on the same pair of shifts,
\begin{equation}
\label{eq:K2D}
\Ksc = \tfrac{8}{3}\,I \kron I
  \;-\; \tfrac{1}{3}\!\!\sum_{(U_p,U_q) \neq (I,I)}\!\! U_p \kron U_q ,
\end{equation}
nine terms with
\begin{equation}
\label{eq:alpha2D}
\asc = \tfrac{8}{3} + 8\cdot\tfrac{1}{3} = \tfrac{16}{3}.
\end{equation}
As a nodal stencil, Equation~\eqref{eq:K2D} is the bilinear nine-point
Laplacian, and $\asc = \norm{\Ksc}_\infty$ exactly.

The shift decomposition is independent of mesh resolution, against
$\tfrac{9}{16}4^m$ terms at $\alpha = (m^2+2m+8)/3$ for a generic Pauli
expansion of the same operator.

\section{2D Homogeneous Elasticity}
\label{sec:elasticity}

In this section, we extend the above results to 2D homogeneous periodic elasticity. Ordering rows and
columns by displacement component gives the familiar $2 \times 2$ block
\begin{equation}
\label{eq:K-block}
\Kel = \begin{bmatrix}
  \mathbf{K}_{xx} & \mathbf{K}_{xy} \\
  \mathbf{K}_{yx} & \mathbf{K}_{yy}
\end{bmatrix}.
\end{equation}
The blocks separate~\cite{deville2002high, zienkiewicz2013finite} as
\begin{align}
\mathbf{K}_{xx} &= C\!\left(\mathbf{K}_{1D} \kron \mathbf{M}_{1D}
    + \tfrac{1-\nu}{2}\,\mathbf{M}_{1D} \kron \mathbf{K}_{1D}\right),
    \label{eq:Kxx}\\
\mathbf{K}_{yy} &= C\!\left(\tfrac{1-\nu}{2}\,\mathbf{K}_{1D} \kron \mathbf{M}_{1D}
    + \mathbf{M}_{1D} \kron \mathbf{K}_{1D}\right),
    \label{eq:Kyy}\\
\mathbf{K}_{xy} &= \mathbf{K}_{yx}^\top
  = -C\,\tfrac{1+\nu}{2}\,\bigl(\mathbf{G}_{1D} \kron \mathbf{G}_{1D}\bigr),
  \label{eq:Kxy}
\end{align}
where $C = E/(1-\nu^2)$ and
$\mathbf{G}_{1D} = \tfrac{1}{2}\,\mathrm{circ}(-1,\,0,\,1)$ is the antisymmetric
gradient coupling, whose shift decomposition is
\begin{equation}
\label{eq:G1-shift}
\mathbf{G}_{1D} = \tfrac{1}{2}\Scd - \tfrac{1}{2}\Sc.
\end{equation}
The $\mathbf{G}_{1D} \kron \mathbf{G}_{1D}$ coupling is the new
cross-derivative term. Since $\mathbf{G}_{1D}$ is antisymmetric,
$\mathbf{G}_{1D} \kron \mathbf{G}_{1D}$ is symmetric, so
$\mathbf{K}_{xy} = \mathbf{K}_{yx}$ on a periodic homogeneous grid.

Elasticity carries $n_d = 2$ degrees of freedom per node, so $\Kel$ is
$2N^{2} \times 2N^{2}$ and the $2\log_2 N$ spatial qubits no longer address it. One
further qubit is required, the \emph{dof qubit} $|d\rangle$, with $d = 0$
selecting the $x$ displacement and $d = 1$ the $y$ displacement. It is now the
most significant qubit of the system register, following the $x$ and $y$ registers.

Equation~\eqref{eq:K-block} is the sum
of three matrices, a block-diagonal average, a block-diagonal difference, and a
purely off-diagonal part. With $\mathbf{K}_{xy} = \mathbf{K}_{yx}$ each is a
single Kronecker product,
\begin{equation}
\label{eq:K-split}
\begin{split}
\mathbf{K}_{\mathrm{2D,el}} = {}
&\tfrac{1}{2}\begin{bmatrix}
\mathbf{K}_{xx}{+}\mathbf{K}_{yy} & \mathbf{0}\\[2pt]
\mathbf{0} & \mathbf{K}_{xx}{+}\mathbf{K}_{yy}
\end{bmatrix}\\[4pt]
+{}&\tfrac{1}{2}\begin{bmatrix}
\mathbf{K}_{xx}{-}\mathbf{K}_{yy} & \mathbf{0}\\[2pt]
\mathbf{0} & \mathbf{K}_{yy}{-}\mathbf{K}_{xx}
\end{bmatrix}\\[4pt]
+{}&\begin{bmatrix}
\mathbf{0} & \mathbf{K}_{xy}\\[2pt]
\mathbf{K}_{xy} & \mathbf{0}
\end{bmatrix},
\end{split}
\end{equation}
each term a single Kronecker product of an $N^2 \times N^2$ spatial block with
$I$, $Z$ and $X$ respectively, so that
\begin{equation}
\label{eq:K-dof}
\Kel = \tfrac{\mathbf{K}_{xx}{+}\mathbf{K}_{yy}}{2} \kron I
           + \tfrac{\mathbf{K}_{xx}{-}\mathbf{K}_{yy}}{2} \kron Z
           + \mathbf{K}_{xy} \kron X .
\end{equation}
This splits $\Kel$ into three \emph{Pauli components}, each an
$N^2 \times N^2$ spatial block tensored with one single-qubit matrix. We can now substitute the shift decompositions of $\mathbf{K}_{1D}$, $\mathbf{M}_{1D}$ and
$\mathbf{G}_{1D}$ into Equation~\eqref{eq:K-dof} to arrive at the following proposition.
\begin{proposition}
\label{prop:lcu}
For the periodic plane-stress $Q_4$ elasticity operator with Young's modulus
$E$ and Poisson ratio $\nu$, the global stiffness matrix $\Kel$ admits
the exact LCU
\begin{equation}
\label{eq:lcu}
\Kel = \sum_{(p,q,r)\in\mathcal{T}} c_{pqr}\;
             U_p^{(x)} \kron U_q^{(y)} \kron \sigma_r^{(d)},
\end{equation}
with $U_p, U_q \in \{I, \Sc, \Scd\}$, $\sigma_r \in \{I, Z, X\}$, and
$\mathcal{T}$ the set of triples carrying a nonzero coefficient, of size
$|\mathcal{T}| = 17$ for every $\nu \in (-1, \tfrac12)$ and every
$N = 2^m$ with $m \ge 2$. The subnormalization is
\begin{equation}
\label{eq:alpha}
\ael(\nu) \;=\; \frac{E\,(33+\nu)}{6\,(1-\nu^{2})},
\end{equation}
constant in $N$.
\end{proposition}

\begin{proof}
Substitute the shift decompositions of $\mathbf{K}_{1D}$, $\mathbf{M}_{1D}$ and
$\mathbf{G}_{1D}$ into the three Pauli components of
Equation~\eqref{eq:K-dof} and merge coefficients landing on the same shift
pair. The $I$ and $Z$ components expand over the nine pairs
$\{I,\Sc,\Scd\}^{\kron2}$, the $X$ component over the four pairs carrying a
shift in both directions, since $\mathbf{G}_{1D}$ has no identity term. Of the
nine $Z$ entries only the four axis pairs survive the antisymmetric
combination $\mathbf{K}_{1D}\kron\mathbf{M}_{1D} -
\mathbf{M}_{1D}\kron\mathbf{K}_{1D}$, giving $9 + 4 + 4 = 17$. The nine pairs
are distinct as operators for $m \ge 2$, where $\Sc \neq \Scd$, so no further
merging occurs, and no coefficient vanishes: the $I$ prefactor carries
$(3{-}\nu)$ and the $Z$ and $X$ prefactors carry $(1{+}\nu)$, vanishing only at
$\nu = 3$ and $\nu = -1$ respectively, neither of which the plane-stress
operator admits. For $\alpha$, the merged coefficients of the spatial factors
sum in magnitude to $16/3$, $4$ and $1$ for the $I$, $Z$ and $X$ components
respectively, and weighting them by the component prefactors $C(3{-}\nu)/4$,
$C(1{+}\nu)/4$ and $C(1{+}\nu)/2$ gives Equation~\eqref{eq:alpha}.
\end{proof}
Table~\ref{tab:lcu} groups the $17$ terms. There are at most four distinct
magnitudes. The eight non-identity shift pairs are the eight stencil
neighbors: four axis pairs such as $(S,I)$ and $(I,S^\dagger)$, and four
diagonal pairs such as $(S,S)$ and $(S^\dagger,S)$.

\begin{table}[H]
\centering
\footnotesize
\setlength{\tabcolsep}{5pt}
\renewcommand{\arraystretch}{1.45}
\caption{The 17 LCU terms of the periodic plane-stress $Q_4$ elasticity
operator at $E = 1$, grouped by coefficient. $S \equiv \Sc$,
$S^\dagger \equiv \Scd$, and the pair $(U^{(x)}_p, U^{(y)}_q)$ names the
stencil offset. Summing the distinct magnitudes with their multiplicities gives
$\ael(\nu)$ of Equation~\eqref{eq:alpha}.}
\label{tab:lcu}
\begin{tabular}{@{}llccr@{}}
\toprule
$\sigma_r$ & stencil offsets $(U^{(x)}_p, U^{(y)}_q)$ & \# & coefficient & $\nu=0.3$ \\
\midrule
$I$ & $(I,I)$, the centre & $1$ & $\dfrac{2(3-\nu)}{3(1-\nu^{2})}$ & $+1.9780$ \\
$I$ & all eight neighbours & $8$ & $-\dfrac{3-\nu}{12(1-\nu^{2})}$ & $-0.2473$ \\
$Z$ & $(I,S)$, $(I,S^\dagger)$ & $2$ & $+\dfrac{1}{4(1-\nu)}$ & $+0.3571$ \\
$Z$ & $(S,I)$, $(S^\dagger,I)$ & $2$ & $-\dfrac{1}{4(1-\nu)}$ & $-0.3571$ \\
$X$ & $(S,S^\dagger)$, $(S^\dagger,S)$ & $2$ & $+\dfrac{1}{8(1-\nu)}$ & $+0.1786$ \\
$X$ & $(S,S)$, $(S^\dagger,S^\dagger)$ & $2$ & $-\dfrac{1}{8(1-\nu)}$ & $-0.1786$ \\
\midrule
\multicolumn{2}{@{}l}{$\ael = \sum_k |c_k|$} & $17$ & $\dfrac{E(33+\nu)}{6(1-\nu^{2})}$ & $6.0989$ \\
\bottomrule
\end{tabular}

\end{table}
The spectral norm is $\norm{\Kel}_2 = 4C$ for every $\nu$ and every $N$, so the
ratio $\ael/\norm{\Kel}_2$ is a closed form,
\begin{equation}
\label{eq:el-tightness}
\frac{\ael}{\norm{\Kel}_2} = \frac{33+\nu}{24} ,
\end{equation}
which is $1.3875$ at $\nu = 0.3$.

The term count $L = 17$ and the closed form $\ael(\nu)$ hold at every
$m \ge 2$, against $\tfrac98 4^m - 1$ terms for an exact Pauli expansion of the
same operator (Equation~\eqref{eq:pauli-counts}).

\section{Two-Phase Microstructure}
\label{sec:microstructure}
We now turn our attention to finite element meshes that carry two materials. We refer to the spatial distribution of the two phases as the \emph{microstructure}.  The periodic two-phase cell is the natural target, because it is the unit cell of computational homogenization~\cite{givois2022qft, liu2024towards, wang2026qafe2,balazi2026quantum}. 

With a second phase, the operator is no longer translation invariant. This leads to two challenges. First, the material is defined element by element, whereas the shift
operators act on nodal registers, so the encoding must relate the two indices;
Section~\ref{sec:mapping} constructs this node-to-element mapping from cyclic
shifts. Second, the element moduli must enter the linear combination as
unitaries. A diagonal modulus matrix is not unitary, and expanding it in a
generic basis would tie $L$ to the microstructure;
Section~\ref{sec:chi_oracle} writes it instead as a combination of the identity
and a material reflection, with coefficients determined by $E_1$ and $E_2$.
Sections~\ref{sec:twophase_poisson} and~\ref{sec:twophase_elasticity} then
assemble the two-phase operators from these pieces.

Figure~\ref{fig:inclusion} shows the primary example
used throughout: a square inclusion of side $N/2$ at the center of the cell.
Each phase carries its own material property. For the Poisson problem, $E_1$
and $E_2$ are the conductivities of the two phases; for elasticity, they are
the two Young's moduli, with a common Poisson ratio $\nu$.
Each element $e$ carries a modulus implied from these two phases,
\begin{equation}
\label{eq:twophase}
E_e = E_2 + (E_1 - E_2)\,\chi_e , \qquad \chi_e \in \{0, 1\},
\end{equation}
 where $\chi_e$ is the microstructure-specific binary field discussed in Section \ref{sec:chi_oracle}.
 
\begin{figure}[H]
\centering
\includegraphics[width=0.25\columnwidth]{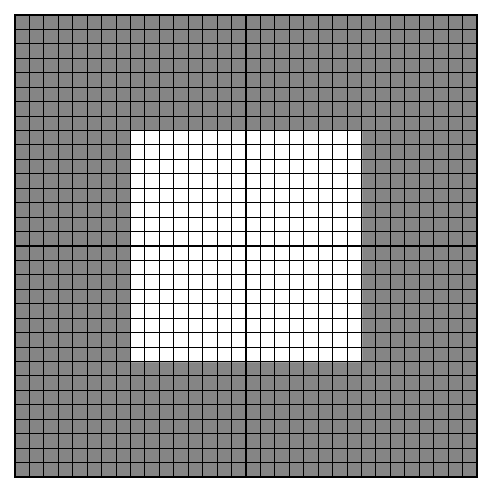}
\caption{The two-phase cell at $m = 5$: a square inclusion of side $N/2$,
centered, at volume fraction $v_f = 1/4$. The inclusion carries $\chi_e = 1$ and
is left unshaded; the matrix carries $\chi_e = 0$. }
\label{fig:inclusion}
\end{figure}

The assumed periodicity of the microstructure links opposite boundaries of the domain. Under this
assumption, an $N \times N$ mesh carries $N^2$ elements and $N^2$ distinct
nodes. We index both nodes and elements using $2m$ qubits, where $N = 2^m$. We
write $|e\rangle$, with $e = (e_x, e_y)$, when the register is read as an
element index, and $|i\rangle$, with $i = (i_x, i_y)$, when read as a nodal
index.

\subsection{Node-to-Element Mapping}
\label{sec:mapping}
 
The material field is defined element by element, whereas the shift operators
of Section~\ref{sec:StructuredEncoding} act on nodal registers. We therefore
define a mapping from one to the other.
 
Each grid node is touched by four adjacent elements, located at the relative
spatial offsets
\begin{equation}
\label{eq:elemoffsets}
(a,b) \in \mathcal{E} := \{(0,0),\,(-1,0),\,(0,-1),\,(-1,-1)\} .
\end{equation}
Figure~\ref{fig:mapping} illustrates this for the node $i = (8,3)$.
 
\begin{figure}[H]
\centering
\includegraphics[width=0.3\columnwidth]{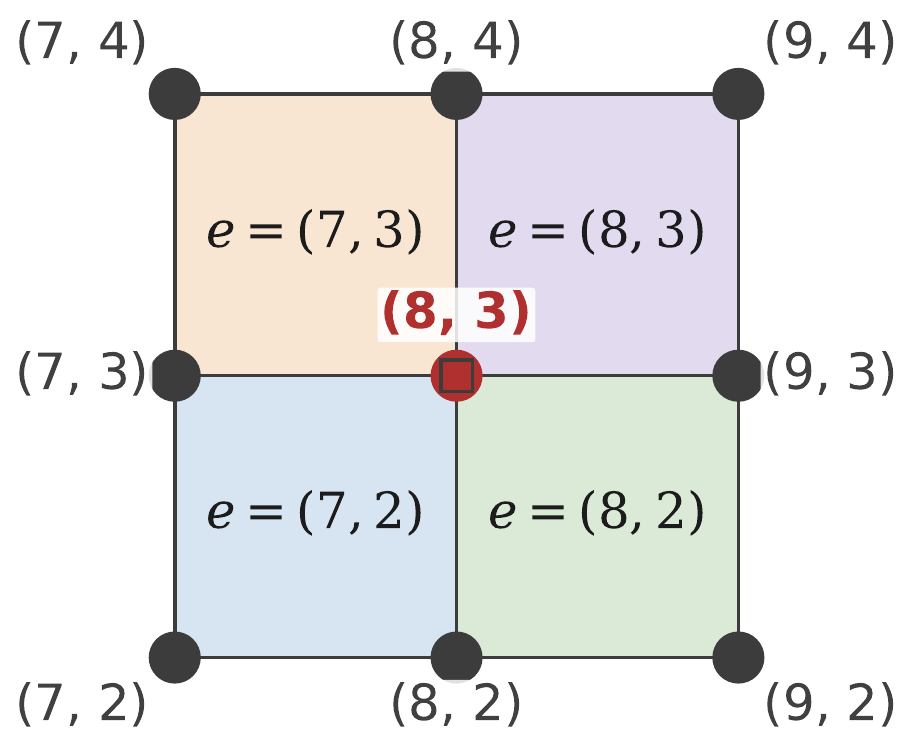}
\caption{The four elements incident on node $i = (8,3)$, on the $2 \times 2$
patch surrounding it. Each element carries the index of its own lower-left
node, marked by a square, so the element at offset $(a,b)$ has index
$i + (a,b)$ and the four indices are $(8,3)$, $(7,3)$, $(8,2)$ and $(7,2)$ as
$(a,b)$ runs over Equation~\eqref{eq:elemoffsets}.}
\label{fig:mapping}
\end{figure}
 
The node-to-element mapping $S_{(a,b)}$ carries the nodal register $|i\rangle$
to the element index $e = \bigl(i + (a,b)\bigr) \bmod N$, componentwise:
\begin{equation}
\label{eq:Sab}
S_{(a,b)} = \Sc^{\,a} \kron \Sc^{\,b} ,
\end{equation}
where $\Sc$ is the 1D cyclic forward shift of Equation~\eqref{eq:Sc-1D} and
$\Sc^{-1} = \Scd$ its backward counterpart. Thus, for example,
\begin{equation}
\label{eq:Sab-example}
S_{(1,0)} = \Sc^{\,1} \kron \Sc^{\,0} = \Sc \kron I .
\end{equation}
Since $\mathcal{E} = \{0,-1\}^2$, the mapping factorizes over the two spatial
registers, and the four offsets of Equation~\eqref{eq:elemoffsets} are
generated by one conditional shift per register.
The offsets are taken modulo $N$, so a node on the last row or column reaches
elements in the first, and $\mathcal{E}$ is the same four offsets at every node
with no exception at the mesh ends. Periodicity is also what makes
$S_{(a,b)}$ a permutation, and hence unitary. $S_{(a,b)}$ acts
purely as a spatial coordinate translation, without referencing the material
property or the microstructure geometry.
 
\subsection{The Material Oracle and Reflection}
\label{sec:chi_oracle}
 
The microstructure itself is defined as a binary field on the elements,
\begin{equation}
\label{eq:chi}
\chi : \{0,\dots,N{-}1\}^2 \to \{0,1\},
\end{equation}
where
\begin{equation}
\label{eq:vf}
v_f = \frac{1}{N^2}\sum_e \chi_e ,
\end{equation}
is the volume fraction of phase~1. We write $r = E_1/E_2$ for the
phase contrast ratio.
 
Two-phase encoding requires a quantum circuit that evaluates the membership
field $\chi$. For the square inclusion in Figure~\ref{fig:inclusion} the
membership test factorizes into independent conditions on $e_x$ and $e_y$. An
element lies inside the inclusion when $e_x \in [N/4,\, 3N/4)$ and
$e_y \in [N/4,\, 3N/4)$. The two leading bits of $e_x$, with $e_{x,m-1}$ the
most significant, place it in one of four quarters, and $e_x$ lies in range
precisely when $e_{x,m-1} \oplus e_{x,m-2} = 1$. Applying the identical
condition to $e_y$ yields the Boolean predicate
\begin{equation}
\label{eq:holePredicate}
\chi_e = \bigl(e_{x,m-1} \oplus e_{x,m-2}\bigr)
\;\wedge\; \bigl(e_{y,m-1} \oplus e_{y,m-2}\bigr) ,
\end{equation}
for every $m \geq 2$. Note that Equation~\eqref{eq:holePredicate} applies only to the microstructure of Figure~\ref{fig:inclusion}; other microstructures require their own predicates.
 
Given the Boolean indicator $\chi_e$, the material phase enters the block
encoding through the \emph{material reflection oracle} $\Rchi$:
\begin{equation}
\label{eq:Rchi}
\Rchi = I - 2\,\mathrm{diag}(\chi),
\end{equation}
which acts on the element state $|e\rangle$ by applying a phase flip
conditional on phase membership,
\begin{equation}
\label{eq:Rchi2}
\Rchi\,|e\rangle = (-1)^{\chi_e}\,|e\rangle .
\end{equation}
The operator $\Rchi$ is diagonal, Hermitian, and unitary ($\Rchi^2 = I$).
Because it marks states with a phase rather than a flag register, it requires
no additional ancilla qubits and enters linear combinations of unitaries
directly. Under Equation~\eqref{eq:holePredicate}, $\Rchi$ acts on exactly
four qubits at every $m$, the two leading bits of each coordinate: two CNOT
gates compute the two parities in place, one controlled $Z$ applies the phase,
and two CNOT gates restore the register. No ancilla is needed, because the
parities are computed in place on the coordinate registers rather than onto
shift-control qubits. For this microstructure the cost is four CNOT and one controlled $Z$ at
every $m \geq 2$.  
 
Combining the mapping $S_{(a,b)}$ with the base oracle $\Rchi$ yields the four
element-wise material reflection operators $R_{(a,b)}$ as seen from node $i$.
Reading right to left, the nodal index is carried to the element index, the
oracle is queried there, and the index is carried back:
\begin{equation}
\label{eq:Rab}
R_{(a,b)} = S_{(a,b)}^{\dagger}\,\Rchi\,S_{(a,b)}
          = I - 2\,\mathrm{diag}\bigl(\chi_{\,i+(a,b)}\bigr)_i .
\end{equation}
Writing $D_{(a,b)} := \mathrm{diag}\bigl(E_{\,i+(a,b)}\bigr)_i$ for the modulus
of the element at offset $(a,b)$ as seen from node $i$, the two-phase field of
Equation~\eqref{eq:twophase} gives
\begin{equation}
\label{eq:twophase_diag}
D_{(a,b)} = \frac{E_1 + E_2}{2}\,I
          \;-\; \frac{E_1 - E_2}{2}\,R_{(a,b)}.
\end{equation}

\section{Two-Phase Poisson}
\label{sec:twophase_poisson}

Section~\ref{sec:2DPoisson} decomposed the homogeneous cell into 1D factors. A
variable modulus destroys that separation, so we work instead at the level of
the element matrix, where the modulus is a scalar multiplier and the shift
structure survives. At unit conductivity the bilinear $Q_4$ element matrix is
\begin{equation}
\label{eq:ke_poisson}
\mathbf{k}^e = \frac{1}{6} \begin{bmatrix}
 4 & -1 & -2 & -1 \\
-1 &  4 & -1 & -2 \\
-2 & -1 &  4 & -1 \\
-1 & -2 & -1 &  4
\end{bmatrix},
\end{equation}
its rows and columns running counterclockwise from the lower-left corner. We
index them by the corners, writing $k^e[c,c']$ with
$c, c' \in \{0,1\}^2$.

For node $(i_x, i_y)$, Figure~\ref{fig:mapping} shows the four incident
elements and their offsets $(a,b) \in \mathcal{E}$. Each of those elements
couples that node to all four of its own nodes, so a second index is needed to
enumerate them. The element at offset $(a,b)$ has lower-left node
$(i_x, i_y) + (a,b)$, so node $(i_x, i_y)$ is its corner $c = -(a,b)$ and its
four nodes are $(i_x, i_y) + (a,b) + c'$ as $c'$ runs over $\{0,1\}^2$. That
node couples to node $(i_x, i_y) + \delta$, with $\delta = (a,b) + c'$, through
the entry $k^e[c, c']$.

An entry in the row at $(i_x, i_y)$ and the column at $(i_x, i_y) + \delta$ is
the operator $\Sc^{-\delta_x} \kron \Sc^{-\delta_y}$, so the element at $(a,b)$
contributes the four-term operator
\begin{equation}
\label{eq:Tab}
\mathbf{T}_{(a,b)} := \!\!\sum_{c' \in \{0,1\}^2}\!\!
  k^e\bigl[c,\,c'\bigr]\;
  \Sc^{-\delta_x} \kron \Sc^{-\delta_y} ,
\end{equation}
with $\delta = (a,b) + c'$ running with the sum. Scaling each element by its
modulus gives the assembled operator
\begin{equation}
\label{eq:Kassembly}
\Ksc^{\chi} = \sum_{(a,b) \in \mathcal{E}} D_{(a,b)}\,\mathbf{T}_{(a,b)} ,
\end{equation}
with $D_{(a,b)}$ defined in Equation~\eqref{eq:twophase_diag}.

The four $\mathbf{T}_{(a,b)}$ merge into the nine-term homogeneous operator of
Equation~\eqref{eq:K2D},
\begin{equation}
\label{eq:Tmerge}
\sum_{(a,b) \in \mathcal{E}} \mathbf{T}_{(a,b)} = \Ksc ,
\end{equation}
and each carries the same four coefficients up to order, hence the same
coefficient sum $\tfrac43$, so the four together sum to $\tfrac{16}{3}$, which
is $\asc$ of Equation~\eqref{eq:alpha2D}.

Substituting Equation~\eqref{eq:twophase_diag} into
Equation~\eqref{eq:Kassembly} splits the operator into a mean part, which
merges by Equation~\eqref{eq:Tmerge}, and a contrast part, which does not:
\begin{equation}
\label{eq:K2phase}
\Ksc^{\chi} = \frac{E_1{+}E_2}{2}\,\Ksc
  \;-\; \frac{E_1{-}E_2}{2} \sum_{(a,b) \in \mathcal{E}}
        R_{(a,b)}\,\mathbf{T}_{(a,b)} .
\end{equation}
The microstructure appears only inside the reflections.
\begin{proposition}
\label{prop:poisson2phase}
Equation~\eqref{eq:K2phase} is an exact LCU of the two-phase periodic $Q_4$
Poisson operator in which every unitary is a cyclic shift pair, optionally
preceded by one material reflection. For every $E_1 \neq E_2$, every $\chi$ and
every $N = 2^m$ with $m \ge 2$ it has
\begin{equation}
\label{eq:poisson-L-alpha}
L = 9 + 16 = 25 ,
\qquad
\asc^{\chi} = \tfrac{16}{3}\max(E_1, E_2) .
\end{equation}
At $E_1 = E_2$ the second family vanishes and $L$ collapses to the $9$ of
Equation~\eqref{eq:K2D}.
\end{proposition}

\begin{proof}
The first family of Equation~\eqref{eq:K2phase} is $\Ksc$, nine terms. The
second is four reflections carrying four shifts each, and the sixteen
combinations occupy sixteen index slots distinct from each other and from the
first family, so nothing merges: $L = 25$. For $\asc^{\chi}$, the first family
contributes $\tfrac12(E_1{+}E_2)\,\asc$ and the second
$\tfrac12|E_1{-}E_2| \sum_{(a,b)} \|\mathbf{T}_{(a,b)}\|$, where
$\|\mathbf{T}_{(a,b)}\|$ denotes the sum of the coefficient magnitudes of
$\mathbf{T}_{(a,b)}$. Each $\|\mathbf{T}_{(a,b)}\|$ is $\tfrac43$, so the
second sum is $\tfrac{16}{3}$, as is $\asc$, giving
$\asc^{\chi} = \tfrac{16}{3}\bigl[(E_1{+}E_2) + |E_1{-}E_2|\bigr]/2$, which is
$\tfrac{16}{3}\max(E_1,E_2)$. The factor $E_1 - E_2$ gives the collapse.
\end{proof}
Material contrast introduces no extra subnormalization penalty beyond
scaling with the stiffer phase's modulus. Switching the soft and stiff phases
gives the same subnormalization.

This subnormalization reaches the smallest value a shift-based LCU can attain,
$\asc^{\chi} = \norm{\Ksc^{\chi}}_\infty$, whenever the mesh contains at least
one node completely surrounded by the stiffer material. Because the matrix
infinity norm $\norm{\Ksc^{\chi}}_\infty$ is the maximum absolute row sum
across all nodes, a single fully stiffened node raises the matrix norm
to its upper limit of $\frac{16}{3}\max(E_1,E_2)$, which equals
$\asc^{\chi}$.

If no such node exists, as in a fine checkerboard pattern where every
node touches a mix of soft and stiff elements, the maximum row sum
$\norm{\Ksc^{\chi}}_\infty$ is smaller than $\frac{16}{3}\max(E_1,E_2)$, making
the bound less tight. For example, an $8 \times 8$ checkerboard with a $3:1$
material contrast yields $\norm{\Ksc^{\chi}}_\infty = 10.667$ against
$\asc^{\chi} = 16$. The encoding remains exact in all cases, but $\asc^{\chi}$
reflects a worst-case material bound.

Neither $L = 25$ nor the circuit varies with the contrast; only the
\textsc{prep} angles and $\asc^{\chi}$ do. An exact Pauli expansion of the
same operator leads to $(2^{m+1}-1)^2$ terms at every contrast for the
microstructure of Figure~\ref{fig:inclusion}
(Equation~\eqref{eq:pauli-counts}).

\section{Two-Phase Elasticity}
\label{sec:twophase_elasticity}

For two-phase elasticity, the construction of Section~\ref{sec:twophase_poisson} applies unchanged once
$k^e[c,c']$ becomes a $2 \times 2$ block. Two things differ, and both belong to
elasticity rather than to the method: a fourth dof component, $iY$, appears, and the
four element contributions at a node no longer share a sign.

The plane-stress $Q_4$ element matrix is $8 \times 8$, two degrees of freedom
per node. Indexing rows and columns by corners as before, $k^e[c,c']$ is now
the $2 \times 2$ block coupling the two displacement components at the
corner-$c$ node to those at the corner-$c'$ node. Each block resolves over the
real orthogonal basis $\{I, Z, X, iY\}$ of the dof register, the first
three as in Section~\ref{sec:elasticity} and the fourth
$iY = \bigl(\begin{smallmatrix}0 & 1\\ -1 & 0\end{smallmatrix}\bigr)$, which is
real and orthogonal and carries the antisymmetric part of each block. Writing
$[\,\cdot\,]_r$ for the component along $\sigma_r$, the element at offset
$(a,b)$ contributes
\begin{equation}
\label{eq:Tab-el}
\mathbf{T}_{(a,b)} := \!\!\sum_{c' \in \{0,1\}^2}\sum_{r}\!\!
  \bigl[k^e[c,\,c']\bigr]_r\;
  \Sc^{-\delta_x} \kron \Sc^{-\delta_y} \kron \sigma_r ,
\end{equation}
with $c = -(a,b)$ and $\delta = (a,b) + c'$ as in
Equation~\eqref{eq:Tab}. The assembly and the
split are then Equations~\eqref{eq:Kassembly} and~\eqref{eq:K2phase} verbatim,
with $D_{(a,b)}$ acting on the spatial registers alone:
\begin{equation}
\label{eq:K2phase-el}
\Kel^{\chi} = \frac{E_1{+}E_2}{2}\,\Kel
  \;-\; \frac{E_1{-}E_2}{2} \sum_{(a,b) \in \mathcal{E}}
        R_{(a,b)}\,\mathbf{T}_{(a,b)} .
\end{equation}

The two properties that carried Section~\ref{sec:twophase_poisson} now read
differently. The merge still holds, $\sum_{(a,b)} \mathbf{T}_{(a,b)} = \Kel$,
the $17$ terms of Proposition~\ref{prop:lcu}. But each $\mathbf{T}_{(a,b)}$ now
carries ten terms rather than four, distributed $2, 3, 3, 2$ over its four
offsets: the offset at the node itself and the diagonal offset carry
$\{I, X\}$, while the two axis offsets carry $\{I, Z, iY\}$. The reflection
family therefore holds $40$ terms.

The $iY$ component is absent from the merged family and present in the
reflection family, and the reason is structural. Merging restores the symmetry
$\mathbf{K}_{xy} = \mathbf{K}_{yx}$ on which the three-component rewriting of
Equation~\eqref{eq:K-dof} depends; with a variable modulus each element is
weighted separately, the shear block loses that symmetry, and $iY$ carries the
antisymmetric remainder. Its total weight is
\begin{equation}
\label{eq:iY-weight}
\sum_{(a,b)} \sum_{\delta} \bigl| [\mathbf{T}_{(a,b)}]_{iY} \bigr|
  = \frac{|1 - 3\nu|}{1 - \nu^{2}} ,
\end{equation}
so the eight $iY$ terms vanish at $\nu = 1/3$ and nowhere else.
\begin{proposition}
\label{prop:elasticity2phase}
Equation~\eqref{eq:K2phase-el} is an exact LCU of the two-phase periodic
plane-stress $Q_4$ operator in which every unitary is a cyclic shift pair
tensored with a single-qubit $\sigma_r \in \{I,Z,X,iY\}$, optionally preceded
by one material reflection. For every $E_1 \neq E_2$, every $\chi$, every
$\nu \in (-1,\tfrac12)$ and every $N = 2^m$ with $m \ge 2$,
\begin{equation}
\label{eq:el-L}
L = 17 + 40 = 57 , \qquad L = 49 \ \text{at } \nu = 1/3 ,
\end{equation}
and the subnormalization is
\begin{equation}
\label{eq:el-alpha}
\alpha^{\chi}_{\mathrm{2D,el}}
  = \min(E_1,E_2)\,A(\nu) \;+\; |E_1 - E_2|\,B(\nu) ,
\end{equation}
with $A(\nu) = \ael$ of Equation~\eqref{eq:alpha} at $E = 1$ and
\begin{equation}
\label{eq:B-nu}
B(\nu) = \frac{69 + 5\nu + 6\,|1 - 3\nu|}{12\,(1 - \nu^{2})} .
\end{equation}
At $E_1 = E_2$ the second family vanishes and $L$ collapses to $17$.
\end{proposition}

\begin{proof}
The counting is that of Proposition~\ref{prop:poisson2phase} with ten terms per
element in place of four, and the eight $iY$ terms carry total weight
$|1-3\nu|/(1-\nu^{2})$ by Equation~\eqref{eq:iY-weight}, vanishing at
$\nu = 1/3$ and nowhere else, which is the drop to $49$. For the
subnormalization, summing magnitudes in Equation~\eqref{eq:K2phase-el} gives
$\tfrac12(E_1{+}E_2)A + \tfrac12|E_1{-}E_2| \sum_{(a,b)}
\|\mathbf{T}_{(a,b)}\|$, and writing
$\tfrac12(E_1{+}E_2) = \min(E_1,E_2) + \tfrac12|E_1{-}E_2|$ collects the two
into Equation~\eqref{eq:el-alpha} with
$B = \bigl(A + \sum_{(a,b)}\|\mathbf{T}_{(a,b)}\|\bigr)/2$. Evaluating that sum
on the plane-stress element gives Equation~\eqref{eq:B-nu}, whose
absolute-value term is half the $iY$ weight of
Equation~\eqref{eq:iY-weight}.
\end{proof}
Written out, Equation~\eqref{eq:el-alpha} reads
\begin{equation}
\label{eq:el-alpha-explicit}
\begin{aligned}
\alpha^{\chi}_{\mathrm{2D,el}} = \frac{1}{6\,(1 - \nu^{2})} \Bigl[\;
  &\tfrac12 (E_1{+}E_2)(33{+}\nu) \\[-2pt]
  &+\, |E_1{-}E_2|\bigl(18 + 2\nu + 3|1{-}3\nu|\bigr) \Bigr] .
\end{aligned}
\end{equation}

As in the scalar cell, the reflections are present at every contrast, so
$L = 57$ and the circuit are fixed and only the moduli move, against an exact
Pauli expansion of $159{,}338{,}497$ terms at $m = 12$ for the microstructure
of Figure~\ref{fig:inclusion} (Section~\ref{sec:generic}).

\section{The Block-Encoding Circuit}
\label{sec:encoding}

\subsection{Registers}
\label{sec:registers}

The system register carries $2m + \lceil\log_2 n_d\rceil$ qubits ($n_d = 1$ for Poisson and $n_d = 2$ for elasticity), ordered as $x$, $y$, then the dof qubit. The
ancilla carries one shift-control qubit $s$, the
$m-1$ clean carry ancillas $a$ of the incrementer and finally, the \textsc{prep} register of $n_{\mathrm{prep}}$ qubits. Table~\ref{tab:cost} collects the resulting cost in closed form.

\begin{table}[h]
\centering
\footnotesize
\setlength{\tabcolsep}{3.5pt}
\renewcommand{\arraystretch}{1.2}
\caption{Cost of the five encodings as a function of $m$, exact at every
$m \ge 2$ and verified against measurement at every $m$ from $2$ to $12$. A
superscript $\chi$ marks a two-phase cell. $L$ is the number of LCU terms,
the four middle columns are qubit counts, and the last is a gate count. Every
cell also carries one shift control qubit, included in the total; it and the
$m-1$ carry ancillas are clean and shared by every shift in the circuit, so
both are counted once. $n_{\mathrm{prep}}$ can exceed $\lceil\log_2 L\rceil$
because the index is factored; see Section~\ref{sec:prep}. The Toffoli count is taken before transpilation; for the transpiled two-qubit count see Figure~\ref{fig:scaling}.}
\label{tab:cost}
\begin{tabular}{@{}lrrrrrr@{}}
\toprule
 & & \multicolumn{4}{c}{qubits} & Toffoli \\
\cmidrule(lr){3-6}
cell & $L$ & system & carry & \textsc{prep} & total & gates \\
\midrule
1D Poisson           & $3$  & $m$    & $m-1$ & $2$ & $2m+2$  & $4m$    \\
2D Poisson           & $9$  & $2m$   & $m-1$ & $4$ & $3m+4$  & $8m$    \\
2D Poisson$^{\chi}$  & $25$ & $2m$   & $m-1$ & $7$ & $3m+7$  & $16m-7$ \\
2D elasticity        & $17$ & $2m+1$ & $m-1$ & $6$ & $3m+7$  & $8m+4$  \\
2D elasticity$^{\chi}$ & $57$ & $2m+1$ & $m-1$ & $9$ & $3m+10$ & $16m-1$ \\
\bottomrule
\end{tabular}
\end{table}

Only two entries depend on the mesh: the system register, at
$2m + \lceil\log_2 n_d\rceil$ qubits, and the $m-1$ carry ancillas, which are
returned to $|0\rangle$ and reused by every shift. Everything the
decomposition controls is constant in $m$: the term count, the \textsc{prep}
width, and the coefficients themselves. The Toffoli count is affine in $m$ at
a slope of $8$ per shifted register, which is the four controlled shifts of a
homogeneous cell and the eight of a two-phase cell at $2m-2$ Toffoli each.

\subsection{The Cyclic Shift}
\label{sec:shiftCircuit}

The shift is the controlled ripple-carry cyclic increment of
Section~\ref{sec:1DPoisson}, drawn in Figure~\ref{fig:increment}: $2m-2$
Toffoli, $m$ CNOT and $m-1$ clean ancillas, transpiling to $11m-10$ two-qubit
gates.

\begin{figure}[H]
\centering
\includegraphics[width=0.5\columnwidth]{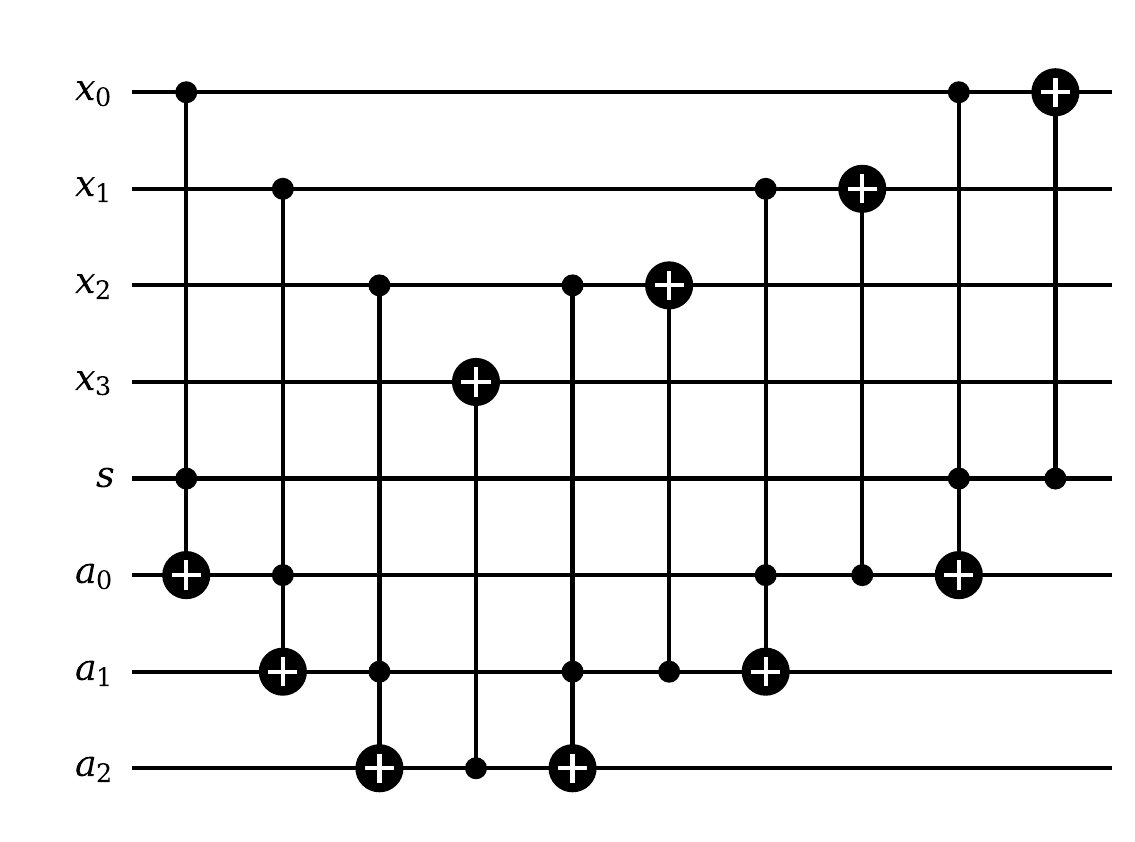}
\caption{The controlled cyclic increment at $m = 4$, the only primitive whose
cost grows with the mesh.}
\label{fig:increment}
\end{figure}

\subsection{PREP}
\label{sec:prep}

Two different preparations are used,
\begin{equation}
\label{eq:preps}
\begin{aligned}
\mathsf{P}_L|0\rangle &= \sum_k \sqrt{|c_k|/\alpha}\;|k\rangle, \\
\mathsf{P}_R|0\rangle &= \sum_k \mathrm{sgn}(c_k)\sqrt{|c_k|/\alpha}\;
   |k\rangle,
\end{aligned}
\end{equation}
so that
\begin{equation}
\label{eq:Ufull}
\Ublk = \bigl(\mathsf{P}_R^\dagger \kron I\bigr)\cdot \mathrm{SEL}
   \cdot \bigl(\mathsf{P}_L \kron I\bigr)
\end{equation}
has top-left block $\sum_k c_k U_k/\alpha = \mathbf{K}/\alpha$ and no signs
appear inside \textsc{select}.

A standard dense \textsc{prep} uses $\lceil \log_2 L\rceil$ qubits. \textsc{select} must then apply each $U_k$ controlled on all $n_{\mathrm{prep}}$ qubits, which requires $L$ controlled
blocks and $\mathcal{O}(L)$ Toffoli gates \cite{babbush2018encoding}. Every
unitary in Equations~\eqref{eq:lcu}, \eqref{eq:K2phase}
and~\eqref{eq:K2phase-el} is, however, a tensor product with one factor per
register: a shift on $x$, a shift on $y$, a Pauli $\sigma_r$ on the dof
register, and at most one reflection $R_{(a,b)}$. Each factor is drawn from a
small fixed set, three shifts per spatial register and four Pauli operators,
so the terms are the combinations of those sets while the operations the
circuit has to realize are their union. We therefore use a \emph{factored} encoding of the
index~\cite{camps2024explicit}, in which \textsc{select} applies one
multiplexer per register and the choice made on one register is independent of
the choice made on the others; for a Pauli LCU, factored encoding over qubits
is standard. Consequently, \textsc{select} cost is independent of $L$; the
Poisson operators at $L=5$ (FD) and $L=9$ (FE) compile to identical
\textsc{select} circuits.
The small price to pay is a wider \textsc{prep} register. The index $k$ is the concatenation of one field per factor: two bits for each shift exponent, two bits for $\sigma_r$, and, for a two-phase operator, three one-bit fields for
the reflection: a flag $g$ that applies $\Rchi$, and bits $a$ and $b$ that
select the offset in Equation~\eqref{eq:Rab}. Unused codes receive zero
amplitude. The register grows from $\lceil \log_2 L\rceil = 4$, $5$, $5$, $6$
to $n_{\mathrm{prep}} = 4$, $7$, $6$, $9$ for the four cells
(Table~\ref{tab:cost}).

\subsection{SELECT}
\label{sec:select}

\textsc{select} applies one controlled operation for each value of each
\textsc{prep} field. The stages, in circuit order, are (see Appendix \ref{sec:circuits}):
\begin{enumerate}\itemsep1pt
  \item \textbf{Spatial:} on each of $x$ and $y$, a controlled increment
    $\Sc$ and a controlled decrement $\Scd$, four controlled shifts in all.
  \item \textbf{Dof:} a controlled $Z$ and $X$, and in the two-phase
    elasticity cell a controlled $iY$, implemented as $R_y(-\pi)$. This stage
    is absent in the scalar cells.
  \item \textbf{Reflection:} in the two-phase cells, decrements on $x$ and $y$
    controlled on $a$ and $b$, $\Rchi$ controlled on $g$, and the matching
    increments. This is $R_{(a,b)}$ of Equation~\eqref{eq:Rab}, and adds four
    controlled shifts and one oracle query.
\end{enumerate}

\textsc{select} therefore contains four controlled shifts in a
homogeneous cell and eight in a two-phase cell, at most three controlled
single-qubit gates, and at most one oracle query, for every $L$. Each
controlled shift costs $2m-2$ Toffoli and each two-bit stage adds two, which
gives the Toffoli counts of Table~\ref{tab:cost}. Since \textsc{select} does not depend on 
$L$, each elasticity circuit exceeds its scalar counterpart by a constant number of
gates.

Finally, every field is at most two bits wide, so each control line is computed by one
Toffoli, and no gate in the four circuits has more than two controls. This
matters for the following reason.

In Qiskit 2.4.2~\cite{javadi2024quantum}, the transpiler synthesizes a
multi-controlled $X$ with four or more controls using idle circuit qubits as
clean ancillas, since a Qiskit circuit assumes every qubit starts in
$|0\rangle$. A block encoding must act correctly on an arbitrary system state,
so the transpiled circuit applies a different operator. The circuits
here avoid this: the factored index of Section~\ref{sec:prep} limits every gate
to two controls, and we synthesize every multi-controlled gate explicitly.
Appendix~\ref{sec:circuits} draws all four encodings as built.

\subsection{Circuit Cost}
\label{sec:where}

Section~\ref{sec:be} noted that the arbitrary-angle rotations of \textsc{prep}
dominate the fault-tolerant cost, each costing
$\mathcal{O}(\log 1/\varepsilon_{\mathrm{syn}})$ $T$ gates while everything else is Clifford and
Toffoli. Table~\ref{tab:where} splits the two-qubit count accordingly: \textsc{prep} is constant in
$m$ and \textsc{select} grows. 

\begin{table}[H]
\centering
\footnotesize
\setlength{\tabcolsep}{4pt}
\renewcommand{\arraystretch}{1.15}
\caption{Cost of the two preparations at $m = 12$, against the whole circuit.
Slots is $2^{n_{\mathrm{prep}}}$, the addressable indices; $L$ of them carry
amplitude and the rest are prepared at zero.}
\label{tab:where}
\begin{tabular}{@{}lrrrrrr@{}}
\toprule
cell & $n_{\mathrm{prep}}$ & slots & $L$ & rot. & \textsc{prep} CX & share \\
\midrule
2D scalar, homog.\ & $4$ & $16$ & $9$ & $29$ & $22$ & $4\%$ \\
2D scalar, 2-phase & $7$ & $128$ & $25$ & $253$ & $240$ & $18\%$ \\
2D elast., homog.\ & $6$ & $64$ & $17$ & $125$ & $114$ & $16\%$ \\
2D elast., 2-phase & $9$ & $512$ & $57$ & $1021$ & $1004$ & $46\%$ \\
\bottomrule
\end{tabular}

\end{table}

Figure~\ref{fig:scaling} gives the transpiled two-qubit count from $m = 2$ to $12$;  the slopes are $48$ gates per qubit of resolution in the homogeneous cells and $96$ in the two-phase cells. The homogeneous
elasticity circuit costs a constant more than the homogeneous scalar one. The intercepts are $25$ and $143$ in the homogeneous cells and $209$ and $1012$ in the two-phase cells. 

\begin{figure}[H]
\centering
\includegraphics[width=0.5\columnwidth]{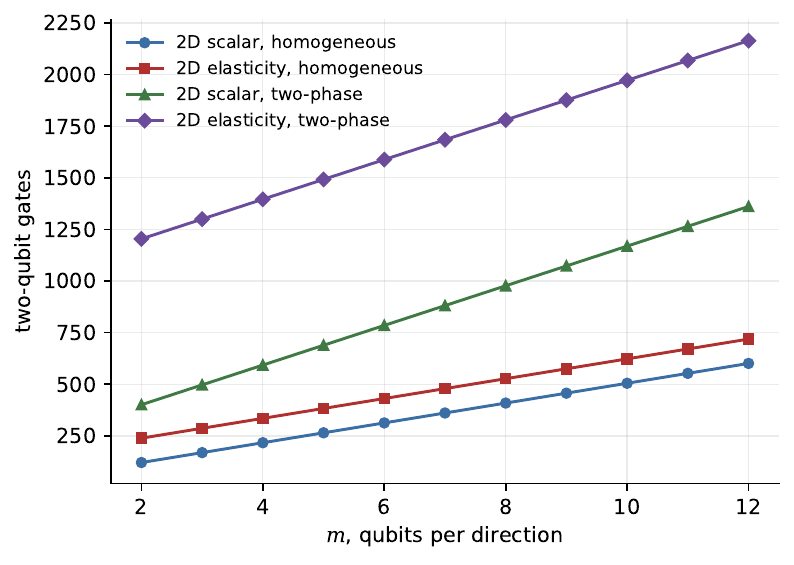}
\caption{Two-qubit gate count against $m$ for the four 2D cells, measured
rather than fitted, for the circuit transpiled to $\{\mathrm{CX}, U\}$ at
optimization level $2$ (Qiskit~2.4.2).}
\label{fig:scaling}
\end{figure}

\subsection{Tightness}
\label{sec:tightness}

The ratio $\alpha/\norm{\mathbf{K}}_2$ factors as
\begin{equation}
\label{eq:tightness-split}
\frac{\alpha}{\norm{\mathbf{K}}_2}
  = \underbrace{\frac{\alpha}{\norm{\mathbf{K}}_\infty}}_{\text{LCU slack}}
    \times
    \underbrace{\frac{\norm{\mathbf{K}}_\infty}{\norm{\mathbf{K}}_2}}
    _{\text{norm gap}} .
\end{equation}
 Table~\ref{tab:tightness} evaluates both factors. Both scalar cells have an
LCU slack of $1$ here: the four element contributions at a node share a sign,
so $\alpha = \norm{\mathbf{K}}_\infty$ whenever some node lies entirely in the
stiffer phase, as shown in Section~\ref{sec:twophase_poisson}; for a
microstructure with no such node, the checkerboard of that section for
instance, the slack exceeds $1$. The elasticity cells have a slack above $1$ at
every microstructure because contributions of opposite sign cancel
(Section~\ref{sec:twophase_elasticity}). For homogeneous elasticity,
$\alpha/\norm{\mathbf{K}}_2 = 1.1935 \times 1.1625 = 1.3875$ at every $m$; in the
two-phase cells the norm gap, and with it the product, decreases as $m$
increases.

\begin{table}[H]
\centering
\footnotesize
\setlength{\tabcolsep}{3.5pt}
\renewcommand{\arraystretch}{1.15}
\caption{The two factors of Equation~\eqref{eq:tightness-split} at $m = 5$,
$\nu = 0.3$, $E_1/E_2 = 3$. The LCU slack is independent of $m$ in all four
cells; the norm gap is independent of $m$ only in the homogeneous cells.
$p_{\mathrm{succ}}$ is the post-selection probability on the maximizing
singular vector.}
\label{tab:tightness}
\begin{tabular}{@{}lrrrr@{}}
\toprule
cell & LCU slack & norm gap & $\alpha/\norm{\mathbf{K}}_2$ & $p_{\mathrm{succ}}$ \\
\midrule
2D scalar, homog.\ & $1.0000$ & $1.3333$ & $1.3333$ & $0.563$ \\
2D scalar, 2-phase & $1.0000$ & $1.3485$ & $1.3485$ & $0.550$ \\
2D elast., homog.\ & $1.1935$ & $1.1625$ & $1.3875$ & $0.519$ \\
2D elast., 2-phase & $1.2473$ & $1.1778$ & $1.4691$ & $0.463$ \\
\bottomrule
\end{tabular}

\end{table}

\section{Implementation}
\label{sec:implementation}

This work is implemented in \pkgname{}, an open-source Python library built
on Qiskit~2.4.2~\cite{javadi2024quantum} and available at \url{\repourl}; the
importable module is \texttt{\libname}. The supported operators are 
\begin{itemize}
    \item \texttt{POISSON1D()}
    \item \texttt{POISSON2D('fd')}
    \item \texttt{POISSON2D('fe')}
     \item \texttt{ELASTICITY2D(nu, E)}
     \item \texttt{POISSON2D\_2PHASE(vf, E1, E2)}
      \item    \texttt{ELASTICITY2D\_2PHASE(nu, vf, E1, E2)}
\end{itemize}

The size is given either as $N$, the grid points per direction, or as $m = \log_2 N$
qubits per direction; the encoded operator carries $N$, $N^2$ or $2N^2$ rows
according to the cell. Listing~\ref{lst:homogeneouselasticity} encodes the homogeneous elasticity cell at $m = 12$. 

\begin{lstlisting}[basicstyle=\ttfamily\scriptsize,
  caption={The homogeneous elasticity cell at $m = 12$. },
  label={lst:homogeneouselasticity}]
from pyblockencode import blockencode, ELASTICITY2D

cell = ELASTICITY2D(nu=0.3)
qc, info = blockencode(cell, N=4096)

print(f"operator = {info.kind}")
print(f"dofs = {info.dofs:,}   L = {info.L}")
print(f"alpha = {info.alpha:.4f}")
print(f"qubits = {info.qubits}")
print(f"Toffoli = {info.toffoli}   CX = {info.cx}")

----- output --------------------------------
operator = elasticity2d
dofs = 33,554,432   L = 17
alpha = 6.0989
qubits = 43
Toffoli = 100   CX = 719
\end{lstlisting}

The two-phase operators place a centered square inclusion of volume fraction
\texttt{vf}; every result reported here uses $v_f = 1/4$, for which the
oracle is Equation~\eqref{eq:holePredicate}. Listing~\ref{lst:two-phase-elasticity} encodes the two-phase elasticity cell at $m = 12$.

\begin{lstlisting}[basicstyle=\ttfamily\scriptsize,
  caption={The two-phase elasticity cell at $m = 12$.},
  label={lst:two-phase-elasticity}]
from pyblockencode import blockencode, ELASTICITY2D_2PHASE

cell = ELASTICITY2D_2PHASE(nu=0.3, vf=0.25, E1=3, E2=1)
qc, info = blockencode(cell, N=4096)

print(f"operator = {info.kind}")
print(f"dofs = {info.dofs:,}   L = {info.L}")
print(f"alpha = {info.alpha:.4f}")
print(f"qubits = {info.qubits}")
print(f"Toffoli = {info.toffoli}   CX = {info.cx}")

----- output --------------------------------
operator = elasticity2d_2phase
dofs = 33,554,432   L = 57
alpha = 19.1209
qubits = 46
Toffoli = 191   CX = 2164
\end{lstlisting}
\section{Conclusions}
\label{sec:conclusions}

In this paper, we presented exact block encodings for 2D homogeneous
elasticity, two-phase 2D Poisson, and two-phase 2D elasticity operators on
periodic finite element grids, built on shift decomposition and a material
reflection oracle. For homogeneous 2D elasticity at $m = 12$ qubits
per direction ($33.6$ million degrees of freedom), the LCU has $17$ terms,
against $18.9$ million for the Pauli expansion.

Two-phase microstructures enter through one material reflection oracle
conjugated by the four node-to-element shifts. The resulting LCUs have
$L = 25$ terms for scalar Poisson and $L = 57$ for elasticity ($49$ at
$\nu = 1/3$), independent of
the grid resolution $N$, the volume fraction $v_f$, and the contrast
$r = E_1/E_2$; the subnormalizations are closed forms in the two moduli,
independent of $N$ and $v_f$.  The formulation is restricted to regular grids, periodic boundary conditions,
and isotropic phases.  

The periodic cell is the setting of computational homogenization that this paper targets. Non-periodic boundaries and unstructured meshes require additional de-periodization and mapping \cite{kharazi2025explicit}, which are not treated here. 

Several extensions follow the same construction: (1) Arbitrary microstructures
require only a different $R_{\chi}$, since neither $L$ nor
$\alpha^{\chi}$ depends on $\chi$, so the open cost is the oracle itself: four
CNOT and one controlled $Z$ for a predicate with a closed form, and a table
lookup for one without. (2) A functionally graded phase replaces the binary
indicator with a discretized modulus, which the same split encodes as one
reflection per bit plane, so $L$ grows with the bit depth rather than with the
number of distinct moduli. (3) The decomposition applies unchanged to any
translation-invariant finite element operator on a regular periodic grid,
including the mass, convection-diffusion and Helmholtz operators, for which
only the element matrix and hence the coefficients change. (4) Extending the work from 2D to 3D 
leave the node-to-element map and the reflection unchanged, with eight incident
elements in place of four, and the scalar cell follows with new constants and
no new argument, while the vector cell requires padding the $3 \times 3$ nodal
blocks into a two-qubit dof register, which is where the additional terms and
subnormalization arise \cite{pechan2026block}. (5) A more efficient \textsc{prep} \cite{suresh2026pyencode} may be desirable, especially in 3D.

Two directions lie outside the present formulation. Nonlinear elasticity
re-encodes,  at every Newton step, a tangent operator whose element moduli vary
with the current state; whether the oracle can be updated rather than rebuilt at
each iteration is open. Finally, the end-to-end cost of a
QSVT-based solve \cite{gilyen2019quantum} remains to be quantified, and the subnormalizations reported here need to be quantified.

\section*{Acknowledgments}
The author acknowledges the Vilas Associate Grant from the University of Wisconsin Graduate School.

\section*{Declarations}
The author declares no conflict of interest.

\section*{Code availability }
The code developed in this work is available at
\url{\repourl}.

\section*{Use of AI Tools}
The author used generative AI assistants, specifically, Claude (Anthropic) and Gemini (Google) during the preparation of this manuscript. The tools were used for code development and drafting of text. All content was reviewed, verified, and edited by the author, who takes full responsibility for the accuracy and integrity of the work.

\appendix
\renewcommand{\thefigure}{A.\arabic{figure}}
\setcounter{figure}{0}
\section{The Four Circuits}
\label{sec:circuits}

Figures~\ref{fig:sc_hom} to~\ref{fig:el_2ph} give the four encodings at
$m = 2$, with only the two preparations boxed and every other gate drawn as
built. The gaps marked \emph{reindex} are where the shift-control qubit $s$ is
cleared, the \textsc{prep} index is recoded for the next term, and $s$ is set
again. 

\begin{figure}[H]
\centering
\includegraphics[width=\textwidth]{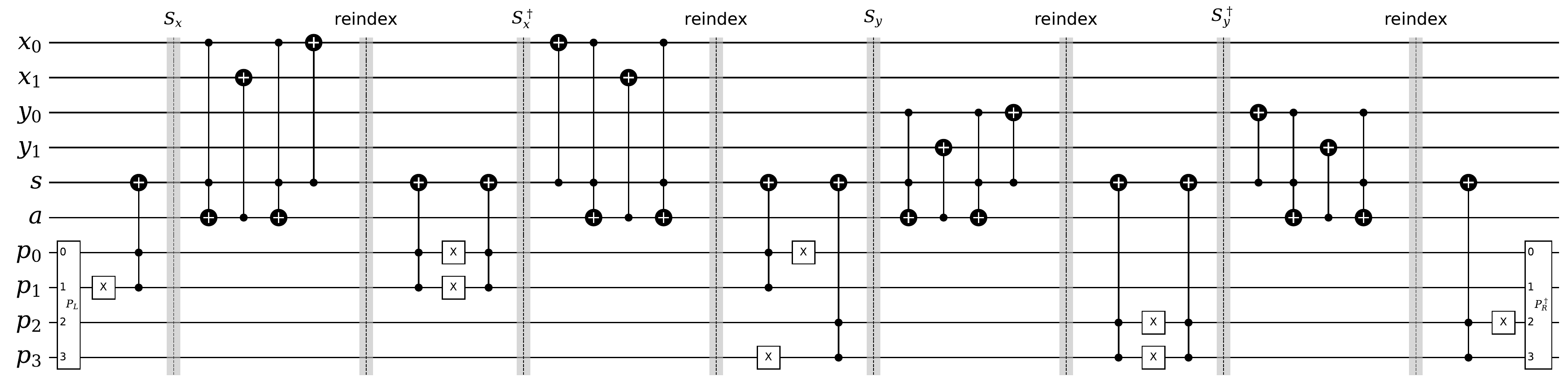}
\caption{Scalar homogeneous cell at $m = 2$, $L = 9$,
$\alpha = 16/3$, on $3m+4 = 10$ qubits. The four bracketed stages
$S_x$, $S_x^\dagger$, $S_y$, $S_y^\dagger$ are the two multiplexers of
Section~\ref{sec:select}, one per spatial register. Within each, the Toffoli
before and after computes and uncomputes the \textsc{select} control from a
two-bit field, and the $X$ gates around it choose which value of that field the
shift controls.}
\label{fig:sc_hom}
\end{figure}

\begin{figure}[H]
\centering
\includegraphics[width=\textwidth]{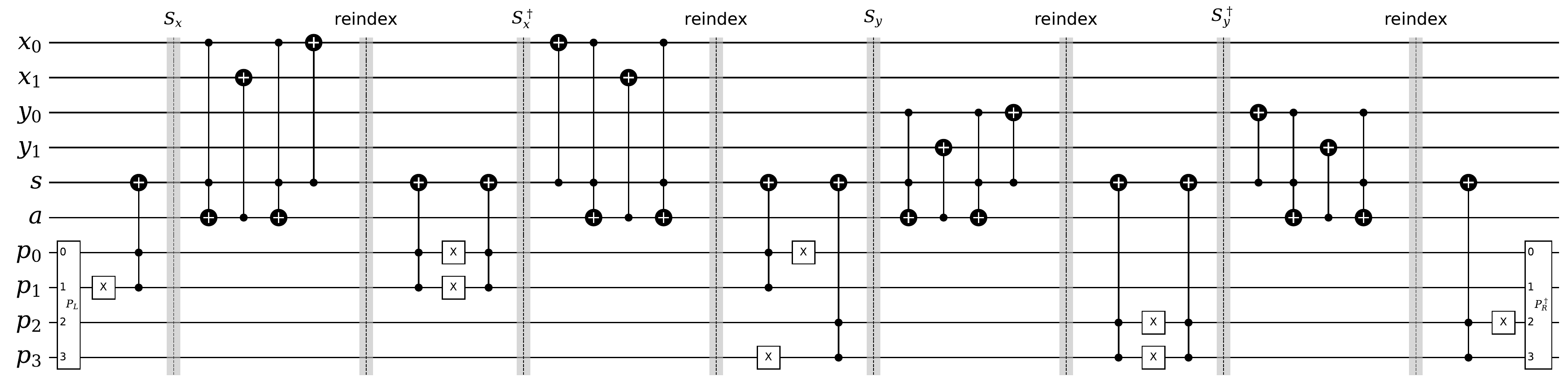}
\caption{Scalar two-phase cell at $m = 2$, $L = 25$, $\alpha = 16$ at
$E_1/E_2 = 3$, on
$3m+7 = 13$ qubits. The first half is Figure~\ref{fig:sc_hom}; the three
brackets $S_{ab}$, $R_\chi$, $S_{ab}^\dagger$ at the right are the conjugation
of Equation~\eqref{eq:Rab} read right to left, a conditional shift on each
spatial register, the oracle controlled on $g$, and the shifts undone. }
\label{fig:sc_2ph}
\end{figure}

\begin{figure}[H]
\centering
\includegraphics[width=\textwidth]{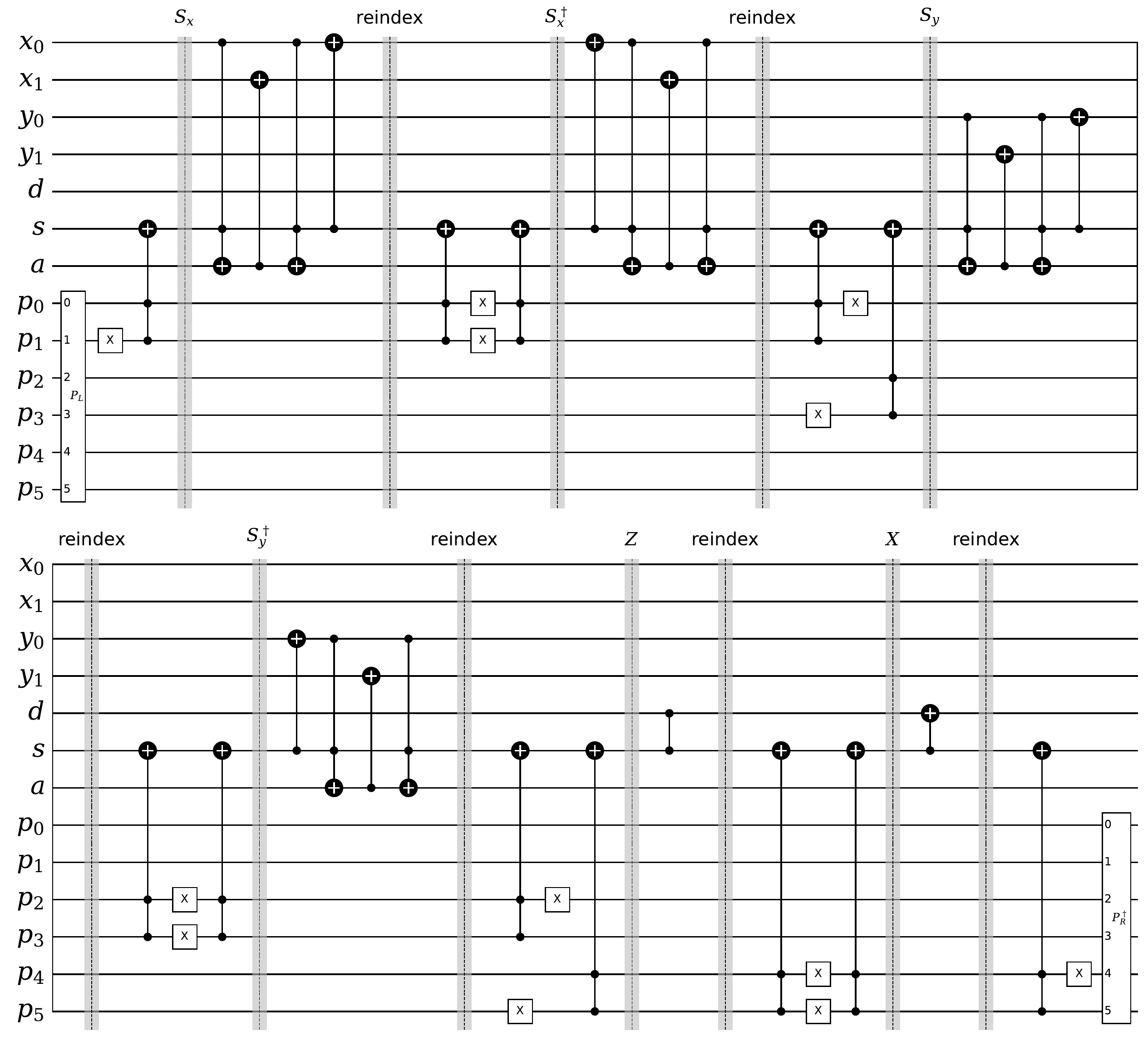}
\caption{Elasticity homogeneous cell at $m = 2$, $L = 17$,
$\alpha = 6.099$ at $\nu = 0.3$, on $3m+7 = 13$ qubits. }
\label{fig:el_hom}
\end{figure}

\begin{figure}[H]
\centering
\includegraphics[width=\textwidth]{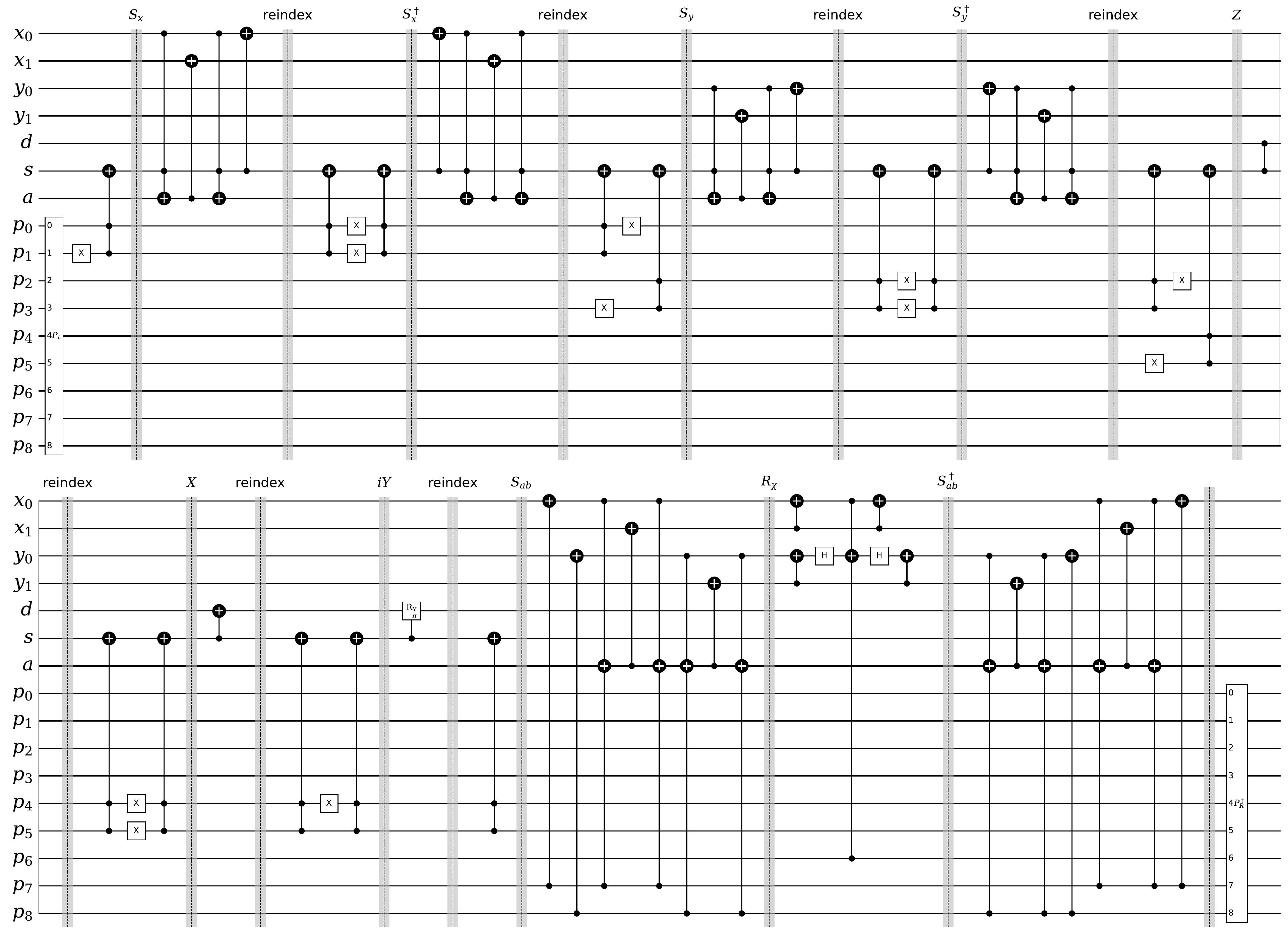}
\caption{Elasticity two-phase cell at $m = 2$, $L = 57$, $\alpha = 19.121$ at
$\nu = 0.3$ and $E_1/E_2 = 3$, on $3m+10 = 16$ qubits. Every stage of the
construction appears: the four spatial multiplexers, the dof
multiplexer $Z$, $X$, $iY$, and the reflection stage, whose three brackets
$S_{ab}$, $R_\chi$, $S_{ab}^\dagger$ are the conjugation of
Equation~\eqref{eq:Rab} read right to left. }
\label{fig:el_2ph}
\end{figure}

\clearpage

\printcredits

\bibliographystyle{unsrt}
\bibliography{ReferencesAppliedQC}

@article{cuccaro2004ripple,
  author  = {Steven A. Cuccaro and Thomas G. Draper and Samuel A. Kutin and
             David Petrie Moulton},
  title   = {A new quantum ripple-carry addition circuit},
  journal = {arXiv preprint quant-ph/0410184},
  year    = {2004}
}

@article{gidney2018halving,
  author  = {Craig Gidney},
  title   = {Halving the cost of quantum addition},
  journal = {Quantum},
  volume  = {2},
  pages   = {74},
  year    = {2018},
  doi     = {10.22331/q-2018-06-18-74}
}

@article{suresh2026pyencode,
  title={PyEncode: an open-source library for structured quantum state preparation},
  author={Suresh, Krishnan and Suresh, Sanjay},
  journal={Quantum Science and Technology},
  year={2026}
}

@article{lu2026memory,
  title={Memory-, Circuit-, and Ansatz-Efficient VQLS for CFD on Hybrid Quantum-HPC Systems},
  author={Lu, Chao and Meena, Muralikrishnan Gopalakrishnan and Perez, Eduardo Antonio Coello and Gottiparthi, Kalyana Chakravarthi and Kim, Seongmin},
  journal={arXiv preprint arXiv:2608.09661},
  year={2026}
}

@article{georges2025pauli,
  title     = {Pauli decomposition via the fast {W}alsh--{H}adamard transform},
  author    = {Georges, Timothy N. and Berntson, Bjorn K. and
               S{\"u}nderhauf, Christoph and Ivanov, Aleksei V.},
  journal   = {New Journal of Physics},
  volume    = {27},
  number    = {3},
  pages     = {033004},
  year      = {2025},
  publisher = {IOP Publishing},
  doi       = {10.1088/1367-2630/adb59b}
}

@article{liu2024towards,
  title={Towards quantum computational mechanics},
  author={Liu, Burigede and Ortiz, Michael and Cirak, Fehmi},
  journal={Computer Methods in Applied Mechanics and Engineering},
  volume={432},
  pages={117403},
  year={2024},
  publisher={Elsevier}
}

@article{givois2022qft,
  title={{QFT}-based homogenization},
  author={Givois, Felix and Kabel, Matthias and Gauger, Nicolas},
  journal={arXiv preprint arXiv:2207.12949},
  year={2022}
}

@article{wang2026qafe2,
  title={{QAFE2}: Quantum accelerated multiscale finite element analysis},
  author={Wang, Yiren and Ortiz, Michael and Cirak, Fehmi},
  journal={Computer Methods in Applied Mechanics and Engineering},
  volume={461},
  pages={119204},
  year={2026},
  publisher={Elsevier}
}

@article{balazi2026quantum,
  title={Quantum enhanced numerical homogenization},
  author={Balazi, Lo{\"\i}c and Deiml, Matthias and Peterseim, Daniel},
  journal={arXiv preprint arXiv:2603.28521},
  year={2026}
}

@article{van2000ubiquitous,
  title={The ubiquitous {Kronecker} product},
  author={Van Loan, Charles F},
  journal={Journal of computational and applied mathematics},
  volume={123},
  number={1-2},
  pages={85--100},
  year={2000},
  publisher={Elsevier}
}

@book{deville2002high,
  title={High-order methods for incompressible fluid flow},
  author={Deville, Michel O and Fischer, Paul F and Mund, Ernest H},
  volume={1},
  year={2002},
  publisher={Cambridge university press Cambridge}
}

@article{gnanasekaran2026efficient,
  title={Efficient quantum access model for sparse structured matrices using linear combination of “things”},
  author={Gnanasekaran, Abeynaya and Surana, Amit},
  journal={Physical Review A},
  volume={113},
  number={2},
  pages={022437},
  year={2026},
  publisher={APS}
}

@book{suresh2026appliedqc,
  title     = {Applied Quantum Computing for Engineers:
               Linear Solvers and Optimization},
  author    = {Suresh, Krishnan},
  edition   = {First},
  year      = {2026},  
  publisher = {Independently published},
  note      = {ISBN-13: 979-8244366846}}

@article{draper2000addition,
  title={Addition on a quantum computer},
  author={Draper, Thomas G},
  journal={arXiv preprint quant-ph/0008033},
  year={2000}
}

@article{pesce2021h2zixy,
  title={{H2ZIXY}: Pauli spin matrix decomposition of real symmetric matrices},
  author={Pesce, Rocco Monteiro Nunes and Stevenson, Paul D},
  journal={arXiv preprint arXiv:2111.00627},
  year={2021}
}

@article{hantzko2024tensorized,
  title={Tensorized {Pauli} decomposition algorithm},
  author={Hantzko, Lukas and Binkowski, Lennart and Gupta, Sabhyata},
  journal={Physica Scripta},
  volume={99},
  number={8},
  pages={085128},
  year={2024},
  publisher={IOP Publishing}
}

@inproceedings{koska2024tree,
  title={A tree-approach {Pauli} decomposition algorithm with application to quantum computing},
  author={Koska, Oc{\'e}ane and Baboulin, Marc and Gazda, Arnaud},
  booktitle={ISC High Performance 2024 Research Paper Proceedings (39th International Conference)},
  pages={1--11},
  year={2024},
  organization={Prometeus GmbH}
}

@article{liu2021variational,
  title={Variational quantum algorithm for the {Poisson} equation},
  author={Liu, Hai-Ling and Wu, Yu-Sen and Wan, Lin-Chun and Pan, Shi-Jie and Qin, Su-Juan and Gao, Fei and Wen, Qiao-Yan},
  journal={Physical Review A},
  volume={104},
  number={2},
  pages={022418},
  year={2021},
  publisher={APS}
}

@article{javadi2024quantum,
  title={Quantum computing with Qiskit},
  author={Javadi-Abhari, Ali and Treinish, Matthew and Krsulich, Kevin and Wood, Christopher J and Lishman, Jake and Gacon, Julien and Martiel, Simon and Nation, Paul D and Bishop, Lev S and Cross, Andrew W and others},
  journal={arXiv preprint arXiv:2405.08810},
  year={2024}
}

@article{harrow2009quantum,
  title={Quantum algorithm for linear systems of equations},
  author={Harrow, Aram W and Hassidim, Avinatan and Lloyd, Seth},
  journal={Physical review letters},
  volume={103},
  number={15},
  pages={150502},
  year={2009},
  publisher={APS}
}

@article{babbush2018encoding,
  title={Encoding electronic spectra in quantum circuits with linear {T} complexity},
  author={Babbush, Ryan and Gidney, Craig and Berry, Dominic W and Wiebe, Nathan and McClean, Jarrod and Paler, Alexandru and Fowler, Austin and Neven, Hartmut},
  journal={Physical Review X},
  volume={8},
  number={4},
  pages={041015},
  year={2018},
  publisher={APS}
}

@book{zienkiewicz2013finite,
  title={The finite element method: its basis and fundamentals},
  author={Zienkiewicz, Olgierd Cecil and Taylor, Robert Leroy and Zhu, Jian Z},
  year={2013},
  publisher={Elsevier}
}

@article{childs2017quantum,
  title={Quantum algorithm for systems of linear equations with exponentially improved dependence on precision},
  author={Childs, Andrew M and Kothari, Robin and Somma, Rolando D},
  journal={SIAM Journal on Computing},
  volume={46},
  number={6},
  pages={1920--1950},
  year={2017},
  publisher={SIAM}
}

@article{childs2012hamiltonian,
   title={Hamiltonian simulation using linear combinations of unitary operations},
   author={Childs, Andrew M and Wiebe, Nathan},
   journal={Quantum Information \& Computation},
   volume={12},
   number={11-12},
   pages={901--924},
   year={2012}
}

@article{deiml2024quantum,
  title={Quantum realization of the finite element method},
  author={Deiml, Matthias and Peterseim, Daniel},
  journal={arXiv preprint arXiv:2403.19512},
  year={2024}
}

@article{berry2015simulating,
  title={Simulating {Hamiltonian} dynamics with a truncated Taylor series},
  author={Berry, Dominic W and Childs, Andrew M and Cleve, Richard and Kothari, Robin and Somma, Rolando D},
  journal={Physical review letters},
  volume={114},
  number={9},
  pages={090502},
  year={2015},
  publisher={APS}
}

@article{camps2022fable,
  title   = {{FABLE}: Fast Approximate Quantum Circuits for Block-Encodings},
  author  = {Camps, Daan and Van Beeumen, Roel},
  journal = {2022 IEEE International Conference on Quantum Computing and Engineering (QCE)},
  pages   = {104--113},
  year    = {2022},
  publisher = {IEEE},
  doi     = {10.1109/QCE53715.2022.00029},
  note    = {arXiv:2205.00081}
}

@article{holscher2026end,
  title={End-to-end quantum algorithm for topology optimization in structural mechanics},
  author={H{\"o}lscher, Leonhard and Ahrend, Oliver and Karch, Lukas and L’Estocq, Carlotta and Andreu, Marc Marfany and Stollenwerk, Tobias and Wilhelm, Frank K and Kowalski, Julia},
  journal={Quantum Science and Technology},
  volume={11},
  number={2},
  pages={025029},
  year={2026},
  publisher={IOP Publishing}
}

@article{xu2025decomposition,
  title={Decomposition-free variational quantum linear solver: Application in computational mechanics},
  author={Xu, Yongchun and Hu, Heng},
  journal={Computer Methods in Applied Mechanics and Engineering},
  volume={447},
  pages={118396},
  year={2025},
  publisher={Elsevier}
}

@article{alkadri2025quantum,
  title={A quantum algorithm for the finite element method},
  author={Alkadri, Ahmad M and Kharazi, Tyler D and Whaley, K Birgitta and Mandadapu, Kranthi K},
  journal={arXiv preprint arXiv:2510.18150},
  year={2025}
}

@article{arora2025implementation,
  title={An implementation of the finite element method in hybrid classical/quantum computers},
  author={Arora, Abhishek and Ward, Benjamin M and Oskay, Caglar},
  journal={Finite Elements in Analysis and Design},
  volume={248},
  pages={104354},
  year={2025},
  publisher={Elsevier}
}

@article{Danz_2026,
   title={An Implementation of Quantum Oracles for the Finite Element Method},
   volume={48},
   ISSN={1095-7197},
   url={http://dx.doi.org/10.1137/25M1761513},
   DOI={10.1137/25m1761513},
   number={3},
   journal={SIAM Journal on Scientific Computing},
   publisher={Society for Industrial & Applied Mathematics (SIAM)},
   author={Danz, Sven and Stollenwerk, Tobias and Ciani, Alessandro},
   year={2026},
   month={June}, pages={B386–B419} }

@article{pechan2026block,
  title={Block encoding of the three-dimensional heterogeneous {Poisson} equation with application to fracture flow},
  author={Pechan, Austin and Golden, John and O’Malley, Daniel},
  journal={Physical Review Applied},
  volume={25},
  number={4},
  pages={044038},
  year={2026},
  publisher={APS}
}

@article{hogancamp2026linear,
  title={Linear-combination-of-unitaries decomposition for the {Laplace} operator},
  author={Hogancamp, Thomas and Demirdjian, Reuben and Gunlycke, Daniel},
  journal={Physical Review A},
  volume={114},
  number={1},
  pages={012455},
  year={2026},
  publisher={APS}
}

@article{mahmud2026moment,
  title={Moment-Structured Block Encodings of Periodic Finite-Difference Operators},
  author={Mahmud, Jishnu and Herrman, Rebekah},
  journal={arXiv preprint arXiv:2607.11596},
  year={2026}
}

@article{yano2026quantum,
  title={Quantum framework for parameterizing partial differential equations via diagonal block-encoding},
  author={Yano, Hiroshi and Sato, Yuki},
  journal={arXiv preprint arXiv:2603.01358},
  year={2026}
}

@article{camps2024explicit,
  title   = {Explicit Quantum Circuits for Block Encodings of Certain Sparse Matrices},
  author  = {Camps, Daan and Lin, Lin and Van Beeumen, Roel and Yang, Chao},
  journal = {SIAM Journal on Matrix Analysis and Applications},
  volume  = {45},
  number  = {1},
  pages   = {801--827},
  year    = {2024},
  publisher = {SIAM},
  doi     = {10.1137/22M1484298},
  note    = {arXiv:2203.10236}
}

@article{sunderhauf2024block,
  title   = {Block-encoding structured matrices for data input in quantum computing},
  author  = {S{\"u}nderhauf, Christoph and Campbell, Earl and Camps, Joan},
  journal = {Quantum},
  volume  = {8},
  pages   = {1226},
  year    = {2024},
  doi     = {10.22331/q-2024-01-11-1226},
  note    = {arXiv:2302.10949}
}

@article{kharazi2025explicit,
  title={Explicit block encodings of boundary value problems for many-body elliptic operators},
  author={Kharazi, Tyler and Alkadri, Ahmad M and Liu, Jin-Peng and Mandadapu, Kranthi K and Whaley, K Birgitta},
  journal={Quantum},
  volume={9},
  pages={1764},
  year={2025},
  publisher={Verein zur F{\"o}rderung des Open Access Publizierens in den Quantenwissenschaften}
}

@article{sturm2025efficient,
  title={Efficient and explicit block encoding of finite difference discretizations of the {Laplacian}},
  author={Sturm, Andreas and Schillo, Niclas},
  journal={ACM Transactions on Quantum Computing},
  year={2025},
  publisher={ACM New York, NY}
}

@misc{kuklinski2025structure,
  title         = {Efficient block-encodings require structure},
  author        = {Kuklinski, Parker and Rempfer, Benjamin and Elenewski, Justin and Obenland, Kevin},
  year          = {2025},
  eprint        = {2509.19667},
  archivePrefix = {arXiv},
  primaryClass  = {quant-ph},
  note          = {arXiv:2509.19667}
}

@inproceedings{gilyen2019quantum,
  title={Quantum singular value transformation and beyond: exponential improvements for quantum matrix arithmetics},
  author={Gily{\'e}n, Andr{\'a}s and Su, Yuan and Low, Guang Hao and Chuang, Isaac L},
  booktitle={Proceedings of the 51st Annual ACM SIGACT Symposium on Theory of Computing (STOC)},
  pages={193--204},
  year={2019},
  doi={10.1145/3313276.3316366},
  note={arXiv:1806.01838}
}

@article{costa2022optimal,
  title={Optimal scaling quantum linear-systems solver via discrete adiabatic theorem},
  author={Costa, Pedro CS and An, Dong and Sanders, Yuval R and Su, Yuan and Babbush, Ryan and Berry, Dominic W},
  journal={PRX Quantum},
  volume={3},
  number={4},
  pages={040303},
  year={2022},
  publisher={APS}
}

 
\end{document}